\documentclass[%
superscriptaddress,
nofootinbib,
amsmath,
amssymb,
aps,
floatfix,
onecolumn
]{revtex4-2}

\usepackage{amsmath,amssymb,amsfonts}
\usepackage{mathtools}    
\usepackage{bm}
\usepackage{braket}
\usepackage{stmaryrd}
\usepackage{mathrsfs}
\usepackage{derivative}
\usepackage{fixdif}

\usepackage{enumitem}

\usepackage{amsthm}
\theoremstyle{definition}
\newtheorem{thm}{Theorem}
\newtheorem{prop}[thm]{Proposition}
\newtheorem{lem}[thm]{Lemma}
\newtheorem{cor}[thm]{Corollary}

\newcommand{\Z}{\mathbb{Z}}

\newcommand{\ave}[2]{\Braket{#1}_{#2}}
\newcommand{\Acal}{\mathcal{A}}

\newcommand{\Ucal}{\mathcal{U}}

\newcommand{\low}{\mathrm{R}}
\newcommand{\high}{\mathrm{E}}

\newcommand{\macro}{\mathrm{macro}}
\newcommand{\micro}{\mathrm{micro}}

\newcommand{\eq}{\mathrm{eq}}

\DeclareMathOperator{\Tr}{\mathrm{Tr}}

\DeclareMathOperator{\diam}{\mathrm{diam}}

\newcommand{\ceqq}{\coloneqq}

\newcommand{\EAP}{{\mathrm{EAP}}}
\newcommand{\ESEP}{{\mathrm{ESEP}}}

\DeclarePairedDelimiter{\abs}{\lvert}{\rvert}
\DeclarePairedDelimiter{\norm}{\lVert}{\rVert}

\usepackage{booktabs}
\usepackage{graphicx}
\usepackage{tikz}
\usetikzlibrary{arrows.meta,calc,decorations.pathreplacing,positioning}

\definecolor{eqblue}{RGB}{220,235,247}
\definecolor{neqorange}{RGB}{252,232,209}
\definecolor{controlgreen}{RGB}{50,145,105}
\definecolor{infopurple}{RGB}{112,78,170}
\definecolor{spinred}{RGB}{200,70,65}
\definecolor{spinblue}{RGB}{65,105,180}

\usepackage{booktabs}

\usepackage{hyperref}
\hypersetup{
  colorlinks=true,
  citecolor=magenta,
  linkcolor=blue,
  urlcolor=cyan
}

\begin{document}
\title{Work Extraction Across a Thermodynamic Hierarchy in Quantum Many-Body Systems}
\author{Akihiro Hokkyo}
 \email{hokkyo@cat.phys.s.u-tokyo.ac.jp}
 \affiliation{Department of Physics, University of Tokyo, 7-3-1 Hongo, Bunkyo-ku, Tokyo 113-8654, Japan}%
\author{Masahito Ueda}%
\affiliation{%
Department of Physics, University of Tokyo, 7-3-1 Hongo, Bunkyo-ku, Tokyo 113-8654, Japan}
\affiliation{
 Institute for Physics of Intelligence, University of Tokyo, 7-3-1 Hongo, Bunkyo-ku, Tokyo 113-0033, Japan}
 \affiliation{
 Fundamental Quantum Science Program (FQSP), TRIP Headquarters, RIKEN, Wako 351-0198, Japan
}%

\begin{abstract}
Thermodynamics is operational
in the sense that 
the very concept of thermal equilibrium depends crucially on the choice of observables, 
while the amount of extractable work is defined relative to the allowed operations.
Here we show that isolated quantum many-body states admit a thermodynamic hierarchical structure of observables and allowed operations,
where the same state can be thermal at one level of the hierarchy and athermal at another. 
Enlarging the set of allowed operations therefore renders such hidden athermality a potential resource for work extraction; 
the resulting work gain is bounded in terms of the difference between the entropy densities of the two levels.
This entropy difference takes the form of the mutual-information density. 
We establish the bounds for three extensions of operational access: 
increased spatial resolution, increased duration of control, 
and nonlocal connectivity,
which provide access to position-state Holevo information, 
correlations between neighboring regions and spatially nonlocal correlations, respectively. 
Thus the difference between entropies at different levels of the thermodynamic hierarchy governs the bound on the work gain associated with moving to a less restricted level.
\end{abstract}

\maketitle

\section*{Introduction}
Thermodynamics exhibits a remarkable universality that goes far beyond its original formulation, 
ranging from ultracold atomic gases~\cite{hoUniversalThermodynamicsDegenerate2004,nascimbeneExploringThermodynamicsUniversal2010} 
to black holes~\cite{bekensteinBlackHolesEntropy1973,hawkingParticleCreationBlack1975}.
This broad applicability originates from the operational character of thermodynamics;
once an observational resolution and its compatible class of operations are specified, 
thermodynamics enables the prediction of observable values at thermal equilibrium 
and dictates which operations are prohibited by the second law of thermodynamics.

This operational perspective has shaped the study of thermodynamics in isolated quantum many-body systems.
The emergence of irreversible thermalization from reversible quantum mechanics 
can be traced back to von Neumann~\cite{neumannBeweisErgodensatzesUnd1929}, 
and recent experiments have shed light on thermalization through the lens of local observables~\cite{kaufmanQuantumThermalizationEntanglement2016,closTimeResolvedObservationThermalization2016,neillErgodicDynamicsThermalization2016}. 
Planck's principle~\cite{liebPhysicsMathematicsSecond1999}, 
a formulation of the second law stating that a cyclic operation cannot lower the energy 
of a system initially in equilibrium without heat exchange, 
has been introduced into quantum thermodynamics as the concept of passivity~\cite{puszPassiveStatesKMS1978,lenardThermodynamicalProofGibbs1978}.
While the original version permits arbitrary unitary operations 
that are generally inaccessible in many-body systems,
the principle has subsequently been reexamined under more physically feasible protocols~\cite{tasakiSecondLawThermodynamics2000,goldsteinSecondLawThermodynamics2013,kanekoWorkExtractionSingle2019,babaWorkExtractabilityEnergy2023,hokkyoUniversalUpperBound2025,chibaSecondLawThermodynamics2026}.
In particular, a state that is locally indistinguishable from a Gibbs state---a state in \textit{microscopic thermal equilibrium} (MITE)---has been shown to satisfy Planck's principle 
under finite-time short-range operations~\cite{hokkyoUniversalUpperBound2025}. 
Chiba \emph{et al.} subsequently demonstrated that a broader class of \textit{infinite-observable macroscopic thermal-equilibrium} (iMATE) states
also satisfies Planck's principle under macroscopic operations~\cite{chibaSecondLawThermodynamics2026}, 
which constitute a subclass of those considered in Ref.~\cite{hokkyoUniversalUpperBound2025} by imposing coarse spatial resolution. 
Here iMATE states are thermal with respect to all additive observables, though not necessarily all local ones~\cite{chibaSecondLawThermodynamics2026}.
These results show that thermodynamic
behavior can be recovered within quantum mechanics under different levels of observation and control.

Here, we shift our focus from a single level to the interrelationship between different operational levels, 
each governed by the laws of thermodynamics.
The same state can therefore admit different thermodynamic descriptions:
it may be thermal with respect to a restricted class of observables but athermal with respect to an enlarged one.
The above results for Planck's principle~\cite{hokkyoUniversalUpperBound2025,chibaSecondLawThermodynamics2026}
further suggest, but do not quantify, ``a trade-off between the notion of thermal equilibrium and the class of operations consistent with the second law of thermodynamics''~\cite{chibaSecondLawThermodynamics2026}:
a thermal state at one level may allow work extraction at the enlarged-control level.
This situation is reminiscent of Maxwell's demon: if a memory with no thermodynamic cost is available, 
the stored information can be utilized through measurement-based feedback control as a resource for additional work extraction.
Here, however, no measurement or feedback is required during the protocol.
The relevant information is already encoded in the initial state 
but cannot be utilized within the restricted class of operations.
While previous studies have compared restricted operations with arbitrary global unitaries~\cite{alickiEntanglementBoostExtractable2013,
onuma-kaluWorkExtractionUsing2018,
puliyilThermodynamicSignaturesGenuinely2022,
salviaOptimalLocalWork2023,
castellanoParallelErgotropyMaximum2025}, 
work extraction across distinct physically motivated operational levels of control in interacting many-body systems remains largely unexplored.
We therefore investigate how much additional work can be extracted 
when expanding the permissible operations from one thermodynamic level to another, 
and how this surplus work relates to the underlying thermodynamic descriptions of both levels.

This question has recently become experimentally relevant as
programmable quantum many-body platforms provide increasingly flexible classes of control.
Optical-tweezer platforms offer local addressability and programmable interactions~\cite{browaeysManybodyPhysicsIndividually2020,kaufmanQuantumScienceOptical2021}, 
while atom rearrangement and coherent transport enable dynamically reconfigurable, nonlocal connectivity~\cite{endresAtombyatomAssemblyDefectfree2016,barredoAtombyatomAssemblerDefectfree2016,bluvsteinQuantumProcessorBased2022}.
Such highly sophisticated control can make microscopic information and correlations accessible that are unavailable under ordinary operations.
Determining the extent to which work extraction is constrained in an isolated quantum many-body system is thus of both foundational and practical importance.

In this paper, we show that the extractable work within the enlarged class 
does not exceed the restricted-level bound plus a contribution from the difference between the entropy densities assigned to the state at the two levels
(see Fig.~\ref{fig:hierarchy-map}). 
This difference in entropy densities takes the form of a mutual-information density, 
which is inaccessible at the restricted level but can become accessible under the enlarged class of operations.
We identify three major cases for the additional work extraction: 
increasing spatial resolution by moving from iMATE to MITE,
increasing the duration of short-range operations to access correlations beyond the scale of local MITE,
and introducing nonlocal operations that jointly address spatially separated regions.
In the first case, the entropy difference gives the Holevo information~\cite{holevoBoundsQuantityInformation1973} between a classical position label and the corresponding local quantum state.
For longer-time short-range operations, it is given by the total correlation among the smaller cells within a larger region, 
whereas for nonlocal operations, it is given by the total correlation among spatially separated regions.
Our bounds thus lend an operational meaning to the difference between the entropies assigned by two thermodynamic descriptions of the same state 
via work extraction.

\begin{figure}[t]
\centering
\includegraphics[width=\linewidth]{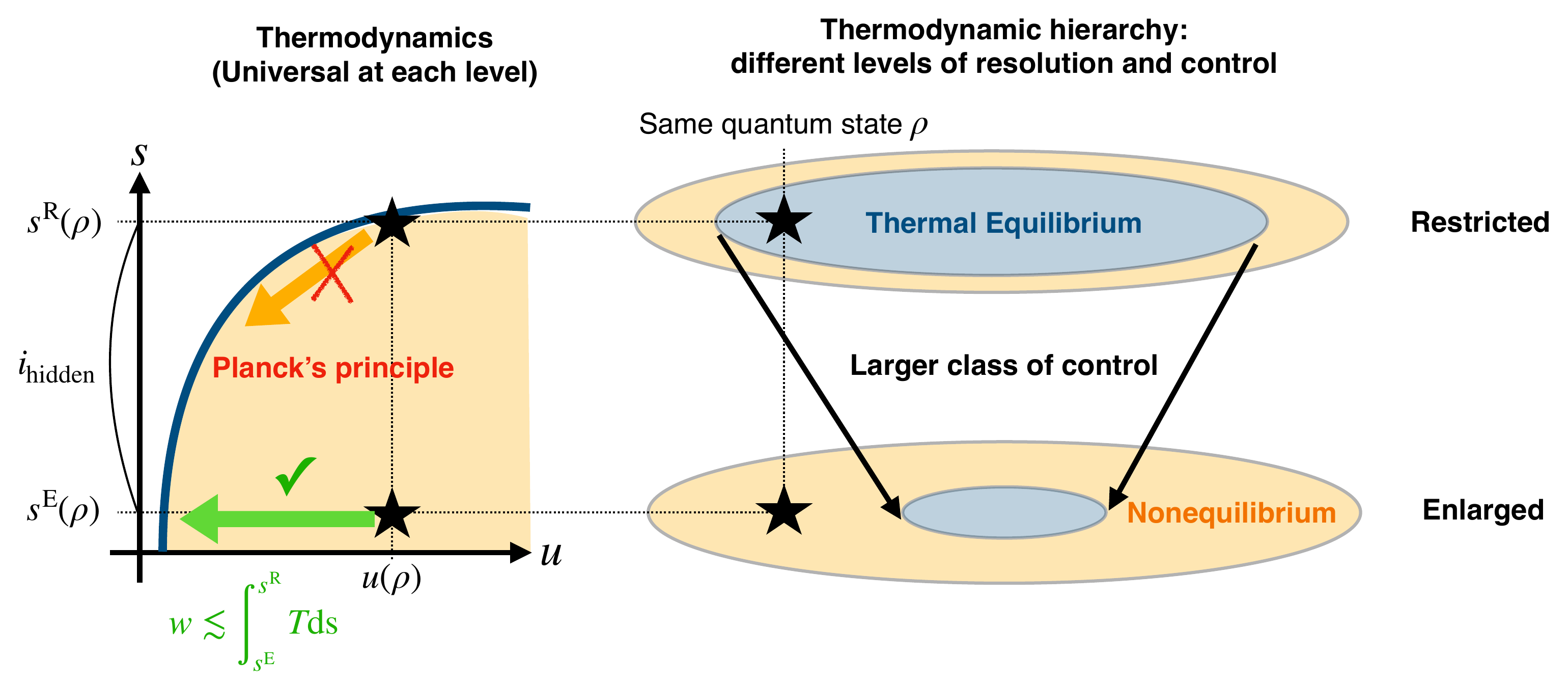}
\caption{\textbf{Schematic of extractable work across thermodynamic levels.}
Consider a quantum state $\rho$ that is in thermal equilibrium 
at a specific thermodynamic level (the ``restricted'' level), 
i.e., with respect to a given observational resolution and a compatible class of operations.
At this level, the entropy $s^{\low}(\rho)$ equals the thermodynamic entropy, 
and Planck's principle prohibits work extraction from this state via an adiabatic cycle.
With a higher observational resolution and a broader class of operations (the ``enlarged'' level), 
the same state $\rho$ may become out of equilibrium.
Consequently, the expanded class of operations may allow work to be extracted from $\rho$.
Any such additional work is bounded in terms of the entropy difference $s^{\low}(\rho)-s^{\high}(\rho)$,
where $s^{\high}(\rho)$ is the entropy at the enlarged-control level.
This difference can be interpreted as the hidden mutual information $i_{\mathrm{hidden}}$ encoded in $\rho$; 
it is inaccessible at the restricted level but can become accessible through the broader control class.
}
\label{fig:hierarchy-map}
\end{figure}

\section*{Thermodynamic hierarchy}
The term ``thermodynamic hierarchy'' is used here in an operational sense. 
Each hierarchical level is defined by its notion of thermal equilibrium and its class of allowed operations,
thereby determining the amount of work extractable from the system.
At each level, 
thermal equilibrium is defined by a class of observables: 
a state is called thermal 
if the expectation values of such observables are indistinguishable from the thermal values at the same energy.
The class of allowed operations is determined by a subset of quantum operations on the system.
Throughout this paper, 
we consider cyclic operations with no heat exchange as the simplest setting in which thermodynamics emerges.
Under this framework, operations are described by unitary evolution governed by cyclic time-dependent Hamiltonians.
The two classes are required to be compatible: 
a thermal state at a level satisfies Planck's principle under the corresponding operations, that is, 
no extensive work can be extracted from it.
Different levels can therefore provide different thermodynamic descriptions of the same state.
When the classes of observables and operations are enlarged,
a state that is thermal at the restricted-control level need not remain so at the enlarged-control level.
Table~\ref{tab:hierarchy} summarizes the three representative hierarchical levels,
which are arranged from top to bottom 
so that the classes of observables and operations become broader.
Each level is associated with an entropy functional that governs an upper bound on work extraction 
and agrees with thermodynamic entropy for states in equilibrium at that level.
The same state may possess different entropies at different levels;
this mismatch can be exploited as a resource for work extraction.

Specifically, 
observables supported on a subsystem $A$ define MITE on region $A$~\cite{goldsteinThermalEquilibriumMacroscopic2015,goldsteinMacroscopicMicroscopicThermal2017,moriThermalizationPrethermalizationIsolated2018}.
MITE describes observed microscopic thermalization~\cite{kaufmanQuantumThermalizationEntanglement2016,closTimeResolvedObservationThermalization2016,neillErgodicDynamicsThermalization2016}.
In particular, MITE on the length scale $l$
is defined by the class of observables supported on regions of size $l$ or smaller.
Planck's principle for this equilibrium notion was established in Ref.~\cite{hokkyoUniversalUpperBound2025},
where the compatible operations are finite-time short-range operations.
Here, 
by a short-range operation,
we mean a unitary evolution under a time-dependent short-range Hamiltonian.
Note that the same proof applies to a longer $l$-dependent time scale (specified later), 
which we refer to as short time on scale $l$.
We can also consider MITE on nonlocal regions,
where some nonlocal operations become compatible; 
see Supplementary Information for details.
The associated entropy is the von Neumann entropy of the reduced state on that region~\cite{hokkyoUniversalUpperBound2025}.

While MITE captures the limit of accessible length scale under the assumption of single-site resolution,
a weaker notion of thermal equilibrium is required for macroscopic thermalization 
where observers have much coarser resolution.
Indeed, Chiba \textit{et al.}~\cite{chibaSecondLawThermodynamics2026}
defined such macroscopic thermal equilibrium\footnote{
Note that this is different from the so-called MATE~\cite{goldsteinThermalEquilibriumMacroscopic2015,tasakiTypicalityThermalEquilibrium2016,goldsteinMacroscopicMicroscopicThermal2017}, 
which is also an abbreviation for macroscopic thermal equilibrium.
MATE typically treats a finite number of observables without specifying them and
considers the full distribution rather than only the expectation value.
}, termed iMATE, 
using a class of observables given by spatial averages of identical local observables over macroscopic subsystems; 
we refer to them as densities of additive macroscopic observables.
For such states, Planck's principle holds for finite-time macroscopic operations~\cite{chibaSecondLawThermodynamics2026},
which roughly\footnote{
Technically speaking, 
the operations considered in Ref.~\cite{hokkyoUniversalUpperBound2025} are generated by one- and two-body control terms, 
whereas macroscopic operations in Ref.~\cite{chibaSecondLawThermodynamics2026} allow any $O(L^0)$-body interactions.
On the other hand, macroscopic operations in Ref.~\cite{chibaSecondLawThermodynamics2026} assume finite-range interactions, 
which is not assumed in Ref.~\cite{hokkyoUniversalUpperBound2025}.
In this paper, we impose neither of these additional restrictions.
} 
constitute a spatially homogeneous subclass of short-range operations.
At this level, the entropy is calculated after taking the spatial average of reduced states~\cite{chibaSecondLawThermodynamics2026}.

\begin{table}[t]
\centering
\small
\setlength{\tabcolsep}{5pt}
\begin{tabular}{@{}p{0.185\linewidth}p{0.26\linewidth}p{0.32\linewidth}p{0.145\linewidth}@{}}
\toprule
Equilibrium notion & Observable class & Compatible operation class  & Entropy of \\
\midrule
iMATE~\cite{chibaSecondLawThermodynamics2026} & additive macroscopic observables & finite-time macroscopic operations~\cite{chibaSecondLawThermodynamics2026} & spatially averaged reduced state $\overline\rho_l$\\
MITE on scale $l$~\cite{goldsteinThermalEquilibriumMacroscopic2015,goldsteinMacroscopicMicroscopicThermal2017,moriThermalizationPrethermalizationIsolated2018} & observables supported within regions of size $l$ & short-range operations for finite time~\cite{hokkyoUniversalUpperBound2025} or for short time on length scale $l$ & reduced state $\rho_l$\\
MITE on enlarged regions $\{A_i\}$& observables on larger regions & longer-time or nonlocal control & reduced state $\rho_A$\\
\bottomrule
\end{tabular}
\caption{\textbf{Operational thermodynamic hierarchy.} 
Each hierarchical level is defined by the corresponding notion of thermal equilibrium that is specified by a set of observables 
accompanied by a compatible class of operations.
Here, ``compatible'' implies that thermal states satisfy Planck's principle.
At each level, thermal equilibrium is also characterized by the corresponding entropy function.
Moving downward through the table narrows the set of thermal states, 
expands the allowed controls, and does not increase entropy.
}
\label{tab:hierarchy}
\end{table}

\section*{Work Extraction Between Thermodynamic Levels}
To demonstrate our key idea about thermodynamic hierarchy, 
we consider an isolated quantum spin system
on a lattice $\Lambda_L\coloneqq\{-\lceil\frac{L}{2}\rceil+1,\dots\lfloor\frac{L}{2}\rfloor\}^D$ 
with a short-range, $O(L^0)$-body and translation-invariant Hamiltonian $H_L$ and initial state $\rho_L$. 
For a class $\Ucal_L$ of allowed operations, the extractable work density is defined
as
\begin{equation}
 w_L(\rho_L;\Ucal_L)=\frac{1}{|\Lambda_L|}
 \left[
 \Tr(H_L\rho_L)-\inf_{U\in\Ucal_L}
 \Tr\!\left(H_LU\rho_LU^\dagger\right)
 \right].
 \label{eq:work}
\end{equation}
A positive value of this quantity $(w_L=\Theta(L^0))$ implies that the system's energy can be extensively lowered by some operation in $\Ucal_L$; 
otherwise, the state satisfies Planck's principle under this class of operations.
We write the energy density of the initial state $\rho$ as $u^{\eq}$, 
which is assumed to satisfy $u^{\eq}\le \Tr(H_L)/(d^{\abs{\Lambda_L}}\abs{\Lambda_L})$, 
so that the corresponding equilibrium temperature is nonnegative.

We consider work extraction between two levels, 
where the restricted-control level is specified by a class $\Ucal_L^{\low}$ of operations and the enlarged-control level by a class $\Ucal_L^{\high}\supset\Ucal_L^{\low}$.
Let $w_L^{\high}(\rho_L)\ceqq w_L(\rho_L;\Ucal_L^{\high})$ be the work density
available after the control class is enlarged,
which trivially satisfies $w_L^{\high}(\rho_L)\ge w_L^{\low}(\rho_L)\ceqq w_L(\rho_L;\Ucal_L^{\low})$. 
By contrast, our result gives 
upper bounds on $w_L^{\high}(\rho_L)$,
which have the common form
\begin{equation}
 w^{\high}(\rho)
 \leq \overline{w}^{\low}(\rho)
 +\underbrace{\int_{s^{\high}=s^{\low}-i_{\rm hidden}}^{s^{\low}}
 T(s)\,\d s\ ,}_{\mbox{potential work gain from the entropy difference}}
 \quad
 i_{\rm hidden}=s^{\low}-s^{\high}\ (\ge0),
 \label{eq:central}
\end{equation}
where we take the thermodynamic limit $L\to\infty$ and omit the subscript $L$.
Here $\overline{w}^{\low}(\rho)\ge w(\rho;\Ucal^{\low})$ is 
an upper bound on the extractable work density at the restricted-control level,
which vanishes asymptotically for equilibrium states at that level;
$T(s)$ is the thermodynamic temperature as a function of entropy density $s$,
which is obtained from the equilibrium energy density $u(s)$;
$i_{\rm hidden}$ is the mutual-information density unavailable for work extraction at the restricted-control level,
yet potentially available at the enlarged-control level.

\begin{figure}[t]
\centering
\includegraphics[width=\linewidth]{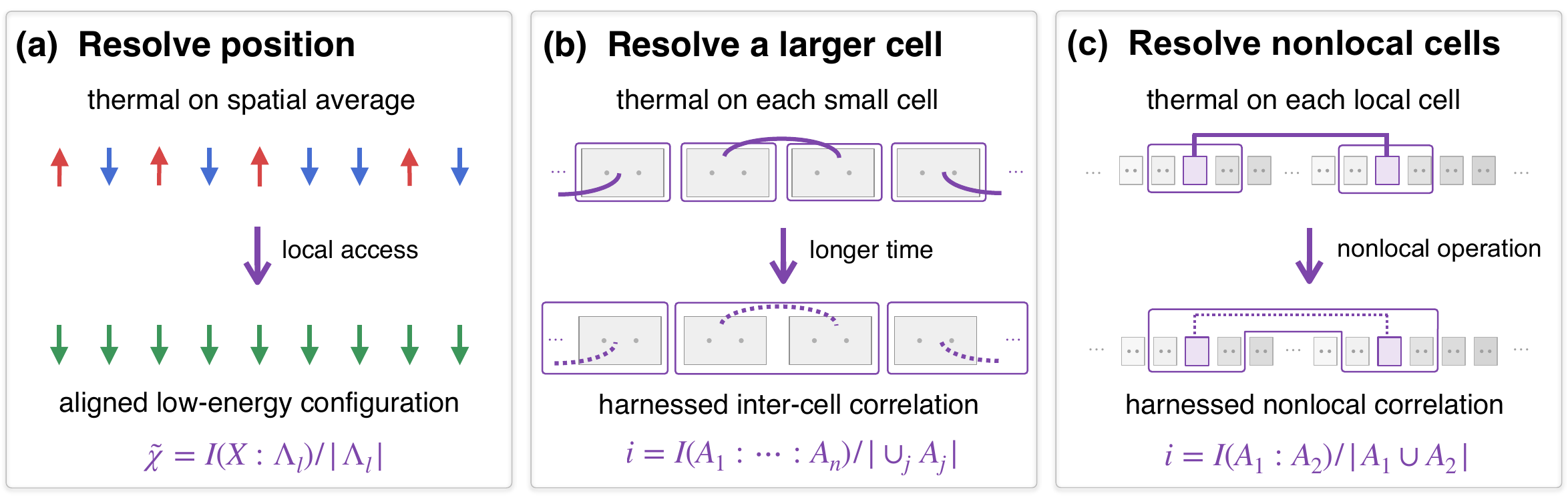}
\caption{\textbf{Three methods to extract work by enlarging the control class.}
(\textbf{a}) A spatially varying configuration can become thermal after positional
averaging, yet local access to microscopic position enables one to utilize the position-state Holevo
information to extract additional work.
(\textbf{b}) Individual small cells appear thermal, although their union need not, owing to inter-cell correlations; 
a longer short-range operation can resolve the larger cell to extract additional work.
(\textbf{c}) Spatially separated cells can each be locally thermal yet jointly
correlated; nonlocal control can resolve their disconnected union as a single
nonlocal cell and access its total correlation to extract additional work.}
\label{fig:three-mechanisms}
\end{figure}

We focus on three concrete cases as summarized in Fig.~\ref{fig:three-mechanisms},
which differ in what the restricted-control level discards: a microscopic position label,
correlations among neighboring regions, or correlations among spatially separated
regions.
To simplify notation, results are presented in the thermodynamic limit $L\to\infty$, 
and the $L$-dependence of certain quantities and operators is omitted.

\subsection{Additional extractable work from microscopic operations}
Let us first consider the work gain from macroscopic thermal equilibrium to MITE.
Both notions are defined for a spatial scale $l(\le L/2)$ as follows.
For every position $x\in\Lambda_L$, 
translate the reduced state around $x$ back to $\Lambda_l$ by a unitary $T_x$
and define the uniform average as 
\begin{equation}
 \rho_{l}^{(x)}=\Tr_{\Lambda_L\setminus \Lambda_l}(T_{x}\rho T_{x}^\dagger),\qquad
 \braket{f_x}_{x\in\Lambda_L}
 =\frac{1}{|\Lambda_L|}\sum_{x\in\Lambda_L}f_x,
 \qquad
 \overline\rho_l=\braket{\rho_l^{(x)}}_{x\in\Lambda_L}.
 \label{eq:posaverage}
\end{equation}
If $\rho_l^{(x)}$ is indistinguishable from an $x$-independent thermal state for every $x$, 
the state is in MITE on the scale $l$.
If only the average $\overline\rho_l$ is thermal, 
the state is in thermal equilibrium only macroscopically;
in the large-$l$ limit, this recovers iMATE for the macroscopically uniform case ($K=1$) in Ref.~\cite{chibaSecondLawThermodynamics2026}.
The corresponding classes of operations are
finite-time short-range operations for MITE 
and macroscopic operations for iMATE.
For simplicity, we only consider spatially homogeneous macroscopic operations, 
which also correspond to the macroscopically uniform case in Ref.~\cite{chibaSecondLawThermodynamics2026}.
This class is contained in a translation-invariant subclass of the short-range operations.

For a state in macroscopic thermal equilibrium, macroscopic control cannot extract extensive work~\cite{chibaSecondLawThermodynamics2026}. 
However, spatially inhomogeneous operations can extract work from the microscopic information embedded in these spatially varying local states $\rho_l^{(x)}$ 
(Fig.~\ref{fig:three-mechanisms} (a)).
Specifically, for the classical--quantum state
$\omega_{X\Lambda_l}=|\Lambda_L|^{-1}\sum_x\ket{x}\!\bra{x}
\otimes\rho_l^{(x)}$, spatial averaging discards precisely the Holevo information density:
\begin{equation}
 \tilde{\chi}_l:=\frac{I(X:\Lambda_l)_\omega}{|\Lambda_l|}
 =\overline s_l
 -s_l,\qquad
 \overline s_l=s_{\Lambda_l}(\overline\rho_l),\qquad
 s_l=\Braket{s_{\Lambda_l}(\rho_l^{(x)})}_{x\in\Lambda_L},
 \qquad s_{\Lambda_l}(\sigma_l)=\frac{S(\sigma_l)}{|\Lambda_l|}.
 \label{eq:holevo-main}
\end{equation}

\begin{thm}[Hierarchical work bound: spatial resolution, informal]\label{thm:spatial-resolution}
Let $w^{\macro}(\rho)$ and $w^{\micro}(\rho)$ denote the extractable work density 
under finite-time macroscopic and short-range operations, respectively.
We then obtain the following bounds for $1\ll l\ll L$:
\begin{align}
 w^{\macro}(\rho)&\le
 \overline{w}^{\macro}(\rho)\coloneqq u^{\eq}-u(\overline s_l),\\
 w^{\micro}(\rho)&\le
 \overline{w}^{\macro}(\rho)
 +\int_{s_l=\overline s_l-\tilde{\chi}_l}^{\overline s_l}
 T(s)\,\d s.
 \label{eq:imate_to_MITE}
\end{align}
\end{thm}

For a state in macroscopic thermal equilibrium, 
$\overline s_l$ represents the thermodynamic entropy in the large-$l$ limit~\cite{chibaSecondLawThermodynamics2026} 
and $\overline{w}^{\macro}=0$.
Meanwhile, the Holevo information density $\tilde{\chi}$ may remain nonzero, 
providing a potential resource for additional work extraction by short-range operations---specifically,
the second term on the right-hand side of Eq.~\eqref{eq:imate_to_MITE}.

A typical example of such work extraction is illustrated in Fig.~\ref{fig:three-mechanisms} (a).
Consider a ferromagnetic spin-$1/2$ system, 
and take a product state as the initial state.
If the direction of each local state is random and $l$ is sufficiently small compared to $L$, 
the state appears thermal on spatial average,
rendering macroscopic operations incapable of extracting extensive work.
This is captured by $\bar{s}_l$, which is close to the maximal value $\ln 2$ for a typical configuration.
On the other hand, we can align each spin via short-range operations,
thereby bringing the state close to the ground state.
In this case, the bound~\eqref{eq:imate_to_MITE} is nearly saturated
with $\tilde{\chi}_l=\ln2$.

\subsection{Additional extractable work from longer-time operations}\label{sec:long-time}
Let us next consider the work gain across different spatial scales $l\ll\hat{l}$ of MITE.
For simplicity, we assume that the interactions of the controlled Hamiltonian $H(t)$ decay faster than $r^{-(D+1+\eta)}\ (\eta>0)$
with respect to the distance $r$;
see Supplementary Information for the general short-range case.
Extending the argument in Ref.~\cite{hokkyoUniversalUpperBound2025} establishes that,
under short-range operations lasting for a short time on scale $l$, 
a state in MITE on scale $l$ continues to satisfy Planck's principle.
Here, a ``short time'' implies that the operation cannot propagate the bulk energy in $\Lambda_l$,
or more precisely $t=o(l/E'_0)$,
where $E'_0$ denotes the scale of the energy density of the controlled Hamiltonian. 
Conversely, increasing $t$ to a time of order $l/E'_0$ or longer can render correlations beyond $l$ accessible for work extraction
(Fig.~\ref{fig:three-mechanisms} (b)). 

To simplify the analysis, we assume that the entropy density of $\rho$ is position independent at each scale, 
and we define $s_l=S(\rho_{\Lambda_l})/|\Lambda_l|$ and
$s_{\hat{l}}=S(\rho_{\Lambda_{\hat{l}}})/|\Lambda_{\hat{l}}|$.
Assuming further that $\hat{l}$ is a multiple of $l$,
the difference in entropy density $i(l\to \hat{l})=s_l-s_{\hat{l}}$ represents the total correlation density among the small cells.
Under these conditions, we establish the following theorem. 
\begin{thm}[Hierarchical work bound: operation time, informal]\label{thm:operation-time}
Let $w^{l}(\rho)$ and $w^{\hat{l}}(\rho)$ denote the extractable work density under short-range operations lasting $t=o(l/E'_0)$ and $t=o(\hat{l}/E'_0)$, respectively.
Then we have the following bounds for $1\ll l\ll \hat{l}\le L/2$:
\begin{align}
 w^{l}(\rho)&\le
 \overline{w}^{l}(\rho)\coloneqq u^{\eq}-u(s_l),\\
 w^{\hat{l}}(\rho)&\le
 \overline{w}^{l}(\rho)
 +\int_{s_{\hat{l}}=s_l-i(l\to \hat{l})}^{s_l}T(s)\,\d s.
 \label{eq:MITE_to_MITE}
\end{align}
\end{thm}

The latter inequality is saturated by the following toy model, 
which is considered in Ref.~\cite{chibaSecondLawThermodynamics2026}. 
On a spin-$1/2$ ring with an even number of sites, 
we choose an odd integer $l$ satisfying $2<l\le L/2$ and 
place singlet Bell pairs $\ket{\psi^-}$ on sites
$k$ and $k+l$ (modulo $L$) for every odd $k$. 
This yields a tensor product of entangled pairs of spins separated by $l$:
\begin{equation}
 \ket{\ESEP_l}=\bigotimes_{k:\,\mathrm{odd}}
 \ket{\psi^-}_{k,k+l},
 \label{eq:eap}
\end{equation}
which we call the equally-separated entangled-pair (ESEP) state;
see Fig.~\ref{fig:eap-structure} for a schematic illustration.
Every interval containing at most $l$ sites is maximally mixed $(s_l=\ln 2)$, 
whereas the entropy density of intervals of length $\hat{l}\gg l$ tends to zero $(s_{\hat{l}}=O(l/\hat{l}))$. 
Consequently, the state is indistinguishable from the infinite-temperature state below the pairing scale $l$, 
yet it contains a resource for work extraction at the scale $\hat{l}$,
which is quantified by the positive difference $s_l-s_{\hat{l}}=i(l\to\hat{l})=\ln 2-O(l/\hat{l})$.

\begin{figure}[t]
\centering
\includegraphics[width=\linewidth]{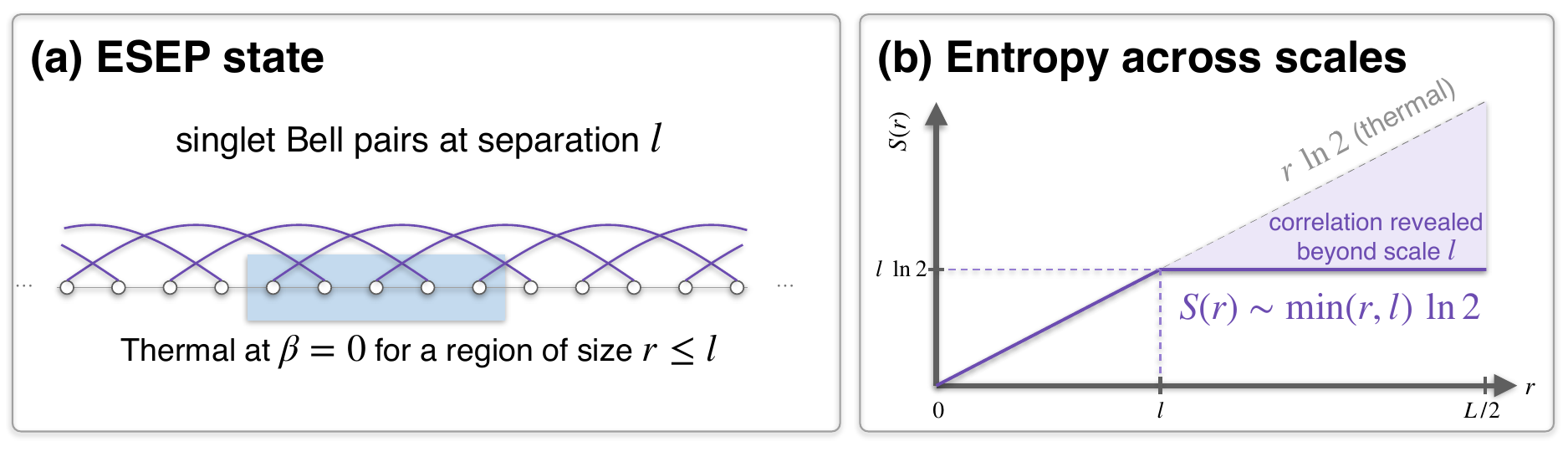}
\caption{\textbf{Equally-separated entangled-pair (ESEP) state and its entropy landscape across varying scales.}
(\textbf{a}) A segment of $\ket{\ESEP_l}$: every site belongs to a singlet Bell
pair whose endpoints are separated by $l$ sites.  
Any contiguous interval of length $r\leq l$ contains no complete pair and is therefore
maximally mixed, appearing thermal at $\beta=0$.
(\textbf{b}) The entropy of a region of length $r$ follows the thermal value
$r\ln2$ up to the pairing scale $l$ and equals $(l+O(1))\ln2$ for
$l\le r\leq L/2$.  
The nonlocal pair correlations give rise to the entropy difference shown in the shaded region.}
\label{fig:eap-structure}
\end{figure}

To calculate a concrete value of work,
we take the frustrated Heisenberg chain for the system Hamiltonian:
\begin{equation}
 H(J_1,J_2)=J_1\sum_j\bm S_j\!\cdot\!\bm S_{j+1}
            +J_2\sum_j\bm S_j\!\cdot\!\bm S_{j+2},
 \qquad J_1,J_2>0.
 \label{eq:j1j2}
\end{equation}
For $l>2$, the state $\ket{\ESEP_l}$ possesses an energy equal to that of the infinite-temperature state, 
which is taken to be zero.
However, a macroscopic operation acting over a timescale of $O(l)/J_1$
maps $\ket{\ESEP_l}$ onto a nearest-neighbor dimer state~\cite{chibaSecondLawThermodynamics2026},
thereby extracting a work density of $3J_1/8$. 
At the Majumdar-Ghosh point $J_2=J_1/2$~\cite{majumdarNextNearestNeighborInteractionLinear1969}, 
this target dimer state becomes the ground state, 
leading to the full saturation of the hierarchy bound~\eqref{eq:MITE_to_MITE} in the limit of $l/\hat{l}\to0$:
\begin{equation}
 w^{\hat{l}}=\int_{0}^{\ln 2} T(s)\, \d s =u(\ln2)-u(0)=\frac{3J_1}{8}.
 \label{eq:saturation}
\end{equation}

\subsection{Additional extractable work from nonlocal operations}
Let us finally consider the work gain from local MITE to nonlocal MITE.
Let $\Acal_l$ be a partition into
connected cells $A$ of linear size $l$.
We take a nonlocal coarser partition
$\widehat{\Acal}_{l}$ whose elements $\hat A$ are unions
\begin{equation}
 \hat A=A_1\cup\cdots\cup A_m,
 \qquad A_\nu\in\Acal_l,
 \label{eq:nonlocal-cluster}
\end{equation}
where $A_i\subset\hat{A}$ may be spatially separated from one another.
We call a state MITE on $\widehat{\Acal}_{l}$ 
if the joint reduced state on every $\hat{A}$ is thermal, 
rather than merely requiring each $\rho_{A_\nu}$ to be thermal.

Short-range operations restricted to short timescales on scale $l$ treat the components of
$\hat A$ as separate thermodynamic subsystems,
thereby failing to extract extensive work from MITE states on scale $l$.
Conversely, nonlocal correlations within each $\hat{A}$ can become accessible for work extraction
if interactions within each $\hat A$ are permitted.
Such nonlocal control can be implemented,
for instance, through direct long-range couplings or by rearranging distant cells
prior to their interaction.  
Accordingly, we define a class of operations compatible with $\widehat{\Acal}_{l}$
via short-time unitary evolutions governed by Hamiltonians that feature interactions 
within each large block $\hat{A}$ of $\widehat{\Acal}_{l}$ 
in addition to those used in short-range operations.

For simplicity, 
we assume that 
the entropy density of $\rho$ for scale $l$ does not depend on position, 
and define $s_l$ as in Sec.~\ref{sec:long-time}.
The key resource for work extraction is the total correlation among the separated cells,
\begin{equation}
 i(\Acal_l\to\widehat{\Acal}_{l})
 =\frac{1}{L^D}\sum_{\hat{A}\in\widehat{\Acal}_{l}}I(A_1:\cdots:A_m)_{\rho_{\hat{A}}}
 =s_l-s(\widehat{\Acal}_{l}),\qquad
 s(\widehat{\Acal}_{l})=\frac{1}{L^D}\sum_{\hat{A}\in\widehat{\Acal}_{l}}S(\rho_{\hat A}),
 \label{eq:ilr}
\end{equation}
where $I(A_1:\cdots:A_m)_{\rho_{\hat{A}}}=D(\rho_{\hat{A}}\|\otimes_{i=1}^m\rho_{A_i})$ denotes the multipartite total correlation.
This is the resource depicted in Fig.~\ref{fig:three-mechanisms} (c).

\begin{thm}[Hierarchical work bound: control locality, informal]
\label{thm:control-locality}
Let $\bar{w}^{l}(\rho)$ be the same as in Thm.~\ref{thm:operation-time}, 
and let $w^{\rm long}(\rho)$ denote the extractable work density under
operations compatible with $\widehat{\Acal}_{l}$.
Then we obtain
\begin{align}
 w^{\rm long}(\rho)&\leq
 \overline w^{l}(\rho)
 +\int_{s(\widehat{\Acal}_{l})=s_l-i(\Acal_l\to\widehat{\Acal}_{l})}^{s_l}T(s)\,\d s.
 \label{eq:long-range-bound}
\end{align}
\end{thm}

In particular, a locally MITE state has $u(s_l)=u^{\eq}$ and hence
$\overline w^{l}=0$, 
while the nonlocal correlation $i(\Acal_l\to\widehat{\Acal}_{l})$ may remain nonzero.  
Direct nonlocal control can exploit such correlations that no connected local region can detect
as a resource for work extraction.

For even $L$, we consider the entangled antipodal-pair state~\cite{chibaExactThermalEigenstates2024}:
\begin{equation}
\ket{\EAP}
=\bigotimes_{j=1}^{L/2}\ket{\psi^-}_{j,j+L/2},
\end{equation}
which is the case of $l=L/2$ in Eq.~\eqref{eq:eap} when $L=2$ (mod $4$).
This state can saturate the bound in Eq.~\eqref{eq:long-range-bound},
as discussed in Ref.~\cite{hokkyoUniversalUpperBound2025}.  
Every connected region smaller than $L/2$ is maximally mixed, 
rendering the state locally indistinguishable from the infinite-temperature state.  
By contrast, for $A'=A+L/2$, the state on $A\cup A'$ is a product of pure singlets.  
Consequently,
the total correlation density attains its maximum value of $\ln2$.
Controlled-NOT operations applied to each $A\cup A'$ 
followed by a single layer of nearest-neighbor unitaries
can then map the state onto the nearest-neighbor dimer state---namely, 
the ground state of the Majumdar-Ghosh Hamiltonian---within a circuit depth of $O(L^0)$.

\section*{Discussion}
The thermodynamic hierarchy introduced herein establishes a rigorous framework for comparing distinct thermodynamic levels 
that differ in the degree of control available to the operator. 
A state that appears thermal at a restricted-control level can become a work resource once the class of operations is enlarged, 
with the corresponding gain in extractable work bounded by the entropy difference between the two descriptions. 
Since this entropy difference is expressed as a mutual-information density, 
the hierarchy imparts an operational thermodynamic significance to information 
that is present within the state yet inaccessible under restricted control. 
The three settings considered here---spatial refinement, 
extended-time control, and nonlocal control---suggest that 
this fundamental relation is not tied to any specific notion of equilibrium, 
but rather reflects a universal structure linking accessible correlations to bounds on work extraction.

An important limitation of the present framework is 
that the initial quantum state is assumed to be fixed,
with the operation optimized for that state. 
The information quantified by the entropy difference is therefore not unknown to the operator; 
rather, it represents information that cannot be utilized within the more restricted class of operations. 
A critical open question is how these bounds are modified when operations must be selected in the absence of complete knowledge regarding the initial state, 
as well as when measurements and measurement-based feedback are incorporated. 
Such incomplete knowledge may stem from imperfect state preparation---an aspect
previously discussed as a possible ingredient for extending Planck's principle to operations 
spanning timescales of order $\Omega(L)$~\cite{chibaSecondLawThermodynamics2026}.

Another compelling avenue for future study is 
whether the thermodynamic hierarchy introduced herein emerges dynamically. 
While we have assigned thermodynamic levels to a given state and evaluated the extractable work under the corresponding classes of operations, 
the dynamics of thermalization for distinct classes of observables are expected to manifest across different length and time scales. 
Consequently, a fundamental question is how this hierarchy of notions of thermal equilibrium is generated dynamically, 
and how such a dynamical hierarchy can be exploited for work extraction.

The framework presented herein could also be extended beyond the current setting. 
In classical many-body systems, a direct analogue of MITE warrants closer inspection, 
as a pure microscopic state cannot be locally thermal in the same sense as a quantum pure state~\cite{goldsteinThermalEquilibriumMacroscopic2015,goldsteinMacroscopicMicroscopicThermal2017}. 
A finite measurement resolution may therefore be necessary to establish the corresponding hierarchy. 
More generally, while we have restricted our treatment to cyclic operations without heat exchange and to work extraction from a single isolated system,
generalizing this operational hierarchy to other thermodynamic settings could elucidate how accessible information and available control can be jointly exploited for further work extraction.

\section*{Acknowledgements}
A.H. thanks Y. Chiba and Y. Yoneta for sharing the construction of the ESEP state prior to the public release of their preprint.
A.H. was supported by KAKENHI Grant No. JP25KJ0833 from the Japan Society for the Promotion of Science (JSPS) 
and by FoPM, a WINGS Program, the University of Tokyo. 
A.H. also acknowledges support from JSR Fellowship, the University of Tokyo. 
This work was supported by KAKENHI Grant No. JP22H01152 from the Japan Society for the Promotion of Science. 
The authors gratefully acknowledge the support from the CREST program ``Quantum Frontiers'' (Grant No. JPMJCR23I1) by the Japan Science and Technology Agency. 
M.U. was supported by the RIKEN TRIP initiative.

\bibliography{../full}

\clearpage

\section*{Supplementary Information}
\setcounter{subsection}{0}
\setcounter{equation}{0}
\renewcommand{\theequation}{S\arabic{equation}}
\setcounter{page}{1}
\renewcommand{\thepage}{S\arabic{page}}
\setcounter{thm}{0}
\renewcommand{\thethm}{S\arabic{thm}}
The three informal work bounds in the main text follow from the finite-size
statements below.  
Proposition~\ref{prop:supp-single-level} bounds the extractable work at one thermodynamic level.  
Combining its homogeneous and position-resolved forms gives Theorem~\ref{thm:supp-spatial-resolution},
which is the formal version of Theorem~\ref{thm:spatial-resolution} in the main text.  
Applying the same proposition at different scales $l$ and $\hat l(\gg l)$ yields
Theorem~\ref{thm:supp-operation-time},
which gives Theorem~\ref{thm:operation-time}.  
Finally, combining the coarse-family bound in
Eq.~\eqref{eq:supp-coarse-single-level} with
Theorem~\ref{thm:supp-control-locality} gives the short-range-to-long-range
bound in Theorem~\ref{thm:control-locality}.  
Throughout the Supplementary Information, 
we set $\hbar=1$, so that energy and inverse time have the same units.

\subsection{Model and notation}

\subsubsection{Lattice and state space}

For $L\in\Z _{>0}$, let
\begin{equation}
 I_L\ceqq
 \left\{-\left\lceil\frac L2\right\rceil+1,\ldots,
 \left\lfloor\frac L2\right\rfloor\right\}
\end{equation}
be the interval of length $L$, and define the
$D$-dimensional hypercubic lattice by $\Lambda_L\ceqq I_L^D$.  
We write $X\Subset A$ when $X$ is a nonempty finite subset of a set $A$.
We assume periodic boundary conditions, 
and 
addition and
subtraction modulo $L$, represented in $\Lambda_L$, are denoted by $+_L$ and
$-_L$, respectively.  We fix an $\ell^p$ norm $|\cdot|$ on $\mathbb R^D$ for $1\le p\le\infty$, 
and define the associated distance on $\Lambda_L$ by
\begin{equation}
 d_L(x,y)
 \ceqq \min_{n\in\Z ^D}|x-y+Ln|
 =|x-_Ly|,
 \label{eq:supp-periodic-distance}
\end{equation}
which corresponds to periodic boundary conditions.
We define the coordinate diameter of $X\Subset\mathbb Z^D$ by
\begin{equation}
 \diam X\ceqq\max_{x,y\in X,1\le i\le D}|x_i-y_i|.
\end{equation}

Each site carries a Hilbert space $\mathcal H_x\cong\mathcal H$ with
$d\ceqq\dim\mathcal H<\infty$.  
For $X\Subset\Z^D$, let
$\mathcal H_X\ceqq\bigotimes_{x\in X}\mathcal H_x$ and
$\mathcal B_X\ceqq\operatorname{End}(\mathcal H_X)$.  The canonical
translation isomorphisms on $\Z ^D$ and $\Lambda_L$ are denoted by
$\tau_a:\mathcal B_X\to\mathcal B_{X+a}$ and 
$\tau_{L,a}:\mathcal B_X\to\mathcal B_{X+_La}$, respectively.  

\subsubsection{Translation-invariant interactions}

An interaction is a family of self-adjoint operators
\begin{equation}
 \Phi=\{\Phi(X)=\Phi(X)^\dagger\in\mathcal B_X:
 X\Subset\Z ^D\}.
\end{equation}
It is translation invariant if
\begin{equation}
 \Phi(X+a)=\tau_a\!\left(\Phi(X)\right)
 \label{eq:supp-ti-interaction}
\end{equation}
for all $X\Subset\Z ^D$ and $a\in\Z ^D$.  
We assume that the interaction $\Phi$ is $k_\Phi$-body:
\begin{equation}
 k_\Phi\ceqq\sup\{|X|:\Phi(X)\ne0\}<\infty,
 \label{eq:supp-k-body}
\end{equation}
where $k_\Phi$ is independent of $L$.

Let $F:[0,\infty)\to[0,\infty)$ be nonincreasing and normalized by
\begin{equation}
 \sum_{z\in\Z ^D}F(|z|)=1.
 \label{eq:supp-F-normalization}
\end{equation}
Consequently,
$\sum_{y\in\Lambda_L}F(d_L(x,y))
=\sum_{z\in \Lambda_L}F(|z|)\le1$ for every $x\in\Lambda_L$.
We assume that the weighted interaction strength of $\Phi$ satisfies
\begin{equation}
 \sum_{\substack{X\Subset\Z^D,\\X\ni x,y}}\frac{\|\Phi(X)\|}{|X|}
 \le E_0 F(|x-y|)
 \label{eq:supp-Phi-locality}
\end{equation}
for all $x,y\in\Z ^D$, 
where $E_0<\infty$ sets the local energy scale.  
In particular, the local energy operator
\begin{equation}
 h_\Phi\ceqq\sum_{\substack{X\Subset\Z^D,\\X\ni0}}\frac{\Phi(X)}{|X|}
 \label{eq:supp-local-energy}
\end{equation}
is norm convergent and obeys $\|h_\Phi\|\le E_0F(0)$.

Whenever the quotient map $\pi_L:\mathbb Z^D\to\Lambda_L$ is injective on
$X\Subset\mathbb Z^D$, we denote the image of
$O_X\in\mathcal B_X$ by $[O_X]_L$.

We call a nonempty
$Y\subseteq\Lambda_L$ $L$-small if it has a lift
$X\Subset\mathbb Z^D$ such that $\pi_L|_X$ is a bijection from $X$ to $Y$ and
$\diam X<L/2$.  Such a lift is unique up to a common translation by an
element of $L\mathbb Z^D$.  
We define
\begin{equation}
 \Phi_L(Y)\ceqq
 \begin{cases}
  [\Phi(X)]_L,&Y\text{ is $L$-small with lift }X,\\
  0,&\text{otherwise}.
 \end{cases}
 \label{eq:supp-periodized-interaction}
\end{equation}
Translation invariance of $\Phi$ makes this definition independent of the
choice of lift.  
The finite-volume Hamiltonian is defined by
\begin{align}
 H_L
 \ceqq\sum_{Y\Subset\Lambda_L}\Phi_L(Y)
  =
 \sum_{a\in\Lambda_L}
 \sum_{\substack{X\Subset\mathbb Z^D;\,0\in X;\\
                   \diam X<L/2}}
 \frac{1}{|X|}\tau_{L,a}\!\left([\Phi(X)]_L\right).
 \label{eq:supp-periodized-Hamiltonian}
\end{align}
This expression also shows that $H_L$ is
translation invariant on $\Lambda_L$.

Note that any translation-invariant two-body interaction can be represented in terms of self-adjoint operators
$h(0)\in \mathcal{B}_{\{0\}}$ and $h(x)\in\mathcal{B}_{\{0,x\}}$
for $x\ne0$,
as
\begin{align}
 \Phi(\{x\})&=\tau_xh(0),\\
 \Phi(\{x,y\})&=\tau_xh(y-x)+\tau_yh(x-y)
 \qquad (x\ne y).
 \label{eq:supp-two-body-recovery}
\end{align}
If these operators satisfy
$\|h(x)\|\le E_*F(|x|)$, then Eq.~\eqref{eq:supp-Phi-locality} holds, for
example, with $E_0=E_*/F(0)$.  

For later use, define the interaction tails
\begin{align}
 r_F(n)&\ceqq\sum_{\substack{z\in\Z ^D\\|z|\ge n+1}}F(|z|),\\
 \bar{r}_F(n)&\ceqq
 \frac{1}{n}\left[r_F(n)+\sum_{m=0}^{n-1}r_F(m)\right]
 \qquad(n\ge1).
 \label{eq:supp-F-tails}
\end{align}
Both functions are nonincreasing and vanish as $n\to\infty$.  Up to a
constant depending only on the lattice geometry, $\bar{r}_F(n)$ bounds the
interaction energy crossing the boundary of a cube of side length $n$, per
unit volume.  If $F$ has a finite first moment,
$\sum_z|z|F(|z|)<\infty$, then $\bar{r}_F(n)\in O(n^{-1})$; 
$F(r)\in O(r^{-(D+1+\eta)})$ for some $\eta>0$ is sufficient for this condition.
On the other hand, 
if $\Phi$ contains a nonzero term supported on more than one site, 
then $\bar{r}_F(n)=\Omega(n^{-1})$.

\subsubsection{Subsystem Hamiltonians and block families}

For $A\subseteq\Lambda_L$, define the subsystem Hamiltonian by retaining only
the interaction terms supported in $A$:
\begin{equation}
 H_{L,A}\ceqq\sum_{Y\Subset A}\Phi_L(Y).
 \label{eq:supp-subsystem-Hamiltonian}
\end{equation}
For a state $\sigma_A$ on $A$, let
$s_A(\sigma_A)\ceqq S(\sigma_A)/|A|$.  The canonical energy density at entropy
density $s\in[0,\ln d]$ is
\begin{align}
 u_{L,A}(s)
 &\ceqq \frac{1}{|A|}\sup_{\beta>0}
 \left\{\beta^{-1}\left[-\ln\Tr_Ae^{-\beta H_{L,A}}+|A|s\right]\right\}
 \nonumber\\
 &=\frac{1}{|A|}
 \min_{\substack{\sigma_A:\,s_A(\sigma_A)\ge s}}
 \Tr(H_{L,A}\sigma_A).
 \label{eq:supp-canonical-energy}
\end{align}
The second equality comes from the Gibbs variational
principle~\cite{wehrlGeneralPropertiesEntropy1978}.
For cubic subsystems, the thermodynamic limit
\begin{equation}
 u(s)\ceqq
 \lim_{L\to\infty}
 u_{L,\Lambda_{\lfloor L/2\rfloor}}(s)
 \label{eq:supp-u-limit}
\end{equation}
exists; 
the proof is given in Sec.~\ref{sec:supp-u-limit}. 
The function $u$ is nondecreasing and convex, 
because the finite-size energy density $u_{L,A}$ has these properties.
We define
\begin{equation}
 T(s)\ceqq \lim_{h\to0+}\frac{u(s+h)-u(s)}{h},\qquad 0\le s<\ln d.
 \label{eq:supp-temperature}
\end{equation}
Hence, we have
\begin{equation}
 u(b)-u(a)=\int_a^bT(s)\,\d s
 \qquad(0\le a\le b\le\ln d).
 \label{eq:supp-u-integral}
\end{equation}

For an integer $1\le l\le L/2$ and $z\in\Lambda_L$, let
\begin{equation}
 Q_{L,l}(z)\ceqq z+_L \Lambda_l
 \label{eq:supp-translated-cube}
\end{equation}
be the translated cube of side length $l$. 
For an offset $a\in\Lambda_L$, we define
\begin{equation}
 \Acal_{L,a}^{(l)}
 \ceqq
 \left\{Q_{L,l}\bigl(a+_L(l n)\bigr):
 n\in\{0,\ldots,\left\lfloor L/l\right\rfloor-1\}^D\right\}.
 \label{eq:supp-disjoint-blocks}
\end{equation}
The fraction of uncovered sites is
\begin{equation}
 n_l(L)\ceqq1-\left(\left\lfloor \frac{L}{l}\right\rfloor\frac{l}{L}\right)^D.
 \label{eq:supp-uncovered-fraction}
\end{equation}
We restrict to admissible scale sequences for which
$n_l(L)=O(L^{-1/2})$.  Any scale sequence with
$L/l_L\to m\in\{2,3,\ldots\}\cup\{\infty\}$ can be changed by a relative
$1+o(L^0)$ correction so that this condition holds.

For any disjoint family $\Acal$, define the volume-weighted average
\begin{equation}
 \ave{f(A)}{A\in\Acal}
 \ceqq
 \frac{1}{|\bigcup\Acal|}\sum_{A\in\Acal}|A|f(A).
 \label{eq:supp-block-average}
\end{equation}
For every function $f$ of a translated cube, direct counting gives
\begin{equation}
 \frac{1}{L^D}\sum_{a\in\Lambda_L}
 \ave{f(A)}{A\in\Acal_{L,a}^{(l)}}
 =\ave{f(Q_{L,l}(z))}{z\in\Lambda_L}\ceqq\frac{1}{L^D}\sum_{z\in\Lambda_L}f(Q_{L,l}(z)).
 \label{eq:supp-counting-identity}
\end{equation}

\subsection{Work extraction and controlled interactions}
We consider an initial state $\rho_L$ on $\Lambda_L$ whose energy corresponds to a nonnegative temperature.
Precisely, we assume that there exist constants
$s^{\eq}\in[0,\ln d]$ and $C<\infty$, independent of $L$, such that
\begin{equation}
 \left|\Tr\left(H_L\rho_L\right)-L^Du(s^{\eq})\right|
 \le CE_0L^{D-1/2}.
 \label{eq:supp-initial-energy}
\end{equation}
We abbreviate $u^{\eq}\ceqq u(s^{\eq})$.
Let $\Ucal_L$ be a class of unitaries that contains the identity.  
The extractable work density is defined as 
\begin{equation}
 w_L(\rho_L;\Ucal_L)
 \ceqq\frac{1}{L^D}
 \left[\Tr(H_L\rho_L)-
 \inf_{U\in\Ucal_L}\Tr(H_LU\rho_LU^\dagger)\right]\ge0.
 \label{eq:supp-work}
\end{equation}

We consider 
a control protocol that is generated on $t\in[0,T_L]$ by 
the cyclic Hamiltonian $K_L(t)$ associated with a time-dependent interaction $\Psi_{L,t}$,
\begin{equation}
 K_L(t)=\sum_{X\Subset\Lambda_L}\Psi_{L,t}(X),
 \qquad \Psi_{L,t}(X)^\dag=\Psi_{L,t}(X),
 \qquad K_L(0)=K_L(T_L)=H_L.
 \label{eq:supp-controlled-Hamiltonian}
\end{equation}
The interaction $\Psi_{L,t}$ may depend on $L$ and need not be translation
invariant.  
We assume that $t\mapsto\Psi_{L,t}(X)$ is piecewise continuous for every $X$.
For a protocol $\Psi=\{\Psi_{L,t}\}_{t\in[0,T_L]}$, let
$U_\Psi(t)$ solve
\begin{equation}
 i\frac{\d}{\d t}U_\Psi(t)
 =K_L(t)U_\Psi(t),\qquad U_\Psi(0)=I.
\end{equation}
For $A,X\subseteq \Lambda_L$, define
\begin{equation}
 \mu_A(X)\ceqq
 \min\{|X\cap A|,|X\cap A^c|\}.
\end{equation}
The boundary action of a protocol across $A|A^c$ is
\begin{equation}
 \mathfrak B_{\Psi,L}(A)
 \ceqq
 \int_0^{T_L}\!\d t
 \sum_{X\Subset \Lambda_L}\mu_A(X)\|\Psi_{L,t}(X)\|.
 \label{eq:supp-boundary-action}
\end{equation}
The protocol is $\epsilon$-compatible with the scale $l$ if
\begin{equation}
 \epsilon_{L,l}(\Psi)\ceqq\max_{z\in\Lambda_L}
 \frac{\mathfrak B_{\Psi,L}(Q_{L,l}(z))}{|Q_{L,l}(z)|}
 \le\epsilon.
 \label{eq:supp-compatibility}
\end{equation}
We define $\mathcal U_L(\epsilon;l)$ as the set of unitaries
$U_\Psi(T_L)$ generated by cyclic protocols satisfying
Eq.~\eqref{eq:supp-compatibility}.
Equation~\eqref{eq:supp-compatibility} means that
the energy that can be transported by $K_L(t)$ through the boundary of $Q_{L,l}(z)$
is negligible compared with the bulk energy.

As a sufficient condition for compatibility,
suppose that there exist a normalized (in the sense of Eq.~\eqref{eq:supp-F-normalization}) nonincreasing function $F'$ and
an energy scale $E_0'$ such that
\begin{equation}
 \sum_{X\Subset\Lambda_L;X\ni x,y}\frac{\|\Psi_{L,t}(X)\|}{|X|}
 \le E_0'F'(d_L(x,y))
 \label{eq:supp-control-locality}
\end{equation}
for any $x,y\in\Lambda_L$,
uniformly in $L$ and $t$.  
We refer to this type of operation as a short-range operation.
For $X\Subset\Lambda_L$, 
set $a_X=|X\cap A|$ and $b_X=|X\cap A^c|$.  Since
\begin{equation}
 \mu_A(X)=\min\{a_X,b_X\}
 \le\frac{2a_Xb_X}{a_X+b_X}
 =\frac{2}{|X|}
 \sum_{x\in X\cap A}\sum_{y\in X\cap A^c}1,
 \label{eq:supp-crossing-multiplicity}
\end{equation}
we have, at every time $t$,
\begin{align}
 \sum_{X\Subset \Lambda_L}\mu_A(X)\|\Psi_{L,t}(X)\|
 &\le2\sum_{x\in A}\sum_{y\notin A}
 \sum_{X\ni x,y}\frac{\|\Psi_{L,t}(X)\|}{|X|}
 \nonumber\\
 &\le2E_0'\sum_{x\in A}\sum_{y\notin A}F'(d_L(x,y)).
 \label{eq:supp-control-boundary-integrand}
\end{align}
Integrating Eq.~\eqref{eq:supp-control-boundary-integrand} over
$t\in[0,T_L]$ proves
\begin{equation}
 \mathfrak B_{\Psi,L}(A)
 \le2E_0'T_L\sum_{x\in A}\sum_{y\notin A}F'(d_L(x,y)).
 \label{eq:supp-control-boundary-bound}
\end{equation}
For a cube of side length $l$, the standard estimate (see Sec.~\ref{sec:supp-boundary-residual-interactions}) gives
\begin{equation}
 \sum_{x\in A}\sum_{y\notin A}F'(d_L(x,y))\le C_D\abs{A}\bar{r}_{F'}(l).
 \label{eq:supp-short-time-compatibility}
\end{equation}
Thus
$E_0'T_L=o(\bar{r}_{F'}(l)^{-1})$ implies
$\epsilon_{L,l}(\Psi)=o(l^0)$; we call this the short-time condition on scale $l$.

A protocol is spatially homogeneous if $\Psi_{L,t}$ is translation invariant
at every $t$.  The homogeneous subclass of $\Ucal_L(\epsilon;l)$ is denoted by
$\Ucal_L^{\mathrm{hom}}(\epsilon;l)$.
In particular, a homogeneous control can be obtained from a time-dependent
translation-invariant interaction $\Psi_t$ on $\Z ^D$ by the same
cutoff and periodization as in Eq.~\eqref{eq:supp-periodized-Hamiltonian}:
\begin{equation}
 K_L(t)=
 \sum_{a\in\Lambda_L}
 \sum_{\substack{X\Subset\Z ^D:\,0\in X\\
                   \diam X<L/2}}
 \frac{1}{|X|}\tau_{L,a}([\Psi_t(X)]_L).
 \label{eq:supp-periodized-control}
\end{equation}
We refer to this type of operation as a spatially homogeneous macroscopic operation.
This class extends the macroscopically uniform macroscopic operations in Ref.~\cite{chibaSecondLawThermodynamics2026}
by dropping the finite-range condition.

Finally, let $\widehat{\Acal}_{L,a}^{(l)}$ be a coarsening of a disjoint small block family
$\Acal_{L,a}^{(l)}$, so that every $\hat{A}\in\widehat{\Acal}_{L,a}^{(l)}$ is a disjoint union of
elements of $\Acal_{L,a}^{(l)}$.  A protocol is compatible with $\widehat{\Acal}_{L,a}^{(l)}$ if
\begin{equation}
 \max_{\hat{A}\in\widehat{\Acal}_{L,a}^{(l)}}
 \frac{\mathfrak B_{\Psi,L}(\hat{A})}{|\hat{A}|}
 \le\epsilon.
 \label{eq:supp-nonlocal-compatibility}
\end{equation}
We denote the corresponding class by
$\Ucal_L(\epsilon;\widehat{\Acal}_{L,a}^{(l)})$,
which includes unitary evolutions under Hamiltonians with interactions 
within each large block $\hat{A}$ of $\widehat{\Acal}_{L,a}^{(l)}$ 
in addition to those used in short-range operations.

\subsection{Work bound at a single thermodynamic level}
We derive upper bounds on work extraction at a single thermodynamic level.
The following argument is essentially the same as in Ref.~\cite{hokkyoUniversalUpperBound2025},
but we make several technical refinements.
For a disjoint family $\Acal$ of subsets of $\Lambda_L$, 
we define 
\begin{equation}
 \bar{w}_L(\rho_L;\Acal)
 \ceqq u^\eq-\ave{u(s_A(\rho_L))}{A\in\Acal},
 \label{eq:supp-fixed-proxy}
\end{equation}
which yields an upper bound on extractable work as follows.

\begin{prop}[Single-level work bound]
\label{prop:supp-single-level}
For $1\le l\le L/2$, we have
\begin{equation}
  w_L(\rho_L;\Ucal_L(\epsilon;l))
 \le
 \bar{w}_L(\rho_L;\Acal_{L,a}^{(l)})
 +E_0O\!\left(\sqrt\epsilon+\bar{r}_F(l)+L^{-1/2}\right),
 \label{eq:supp-single-level-local}
\end{equation}
for any $a\in\Lambda_L$, and
\begin{equation}
  w_L(\rho_L;\Ucal_L^{\mathrm{hom}}(\epsilon;l))
 \le
 u^\eq-u(\bar{s}_l(\rho_L))
 +E_0O\!\left(\sqrt\epsilon+\bar{r}_F(l)+L^{-1/2}\right),
 \label{eq:supp-single-level-hom}
\end{equation}
where
\begin{align}
 \mathcal T_L(\rho_L)
 &\ceqq\frac{1}{L^D}\sum_{z\in\Lambda_L}
   \tau_{L,z}(\rho_L),\\
 \bar{s}_l(\rho_L)
 &\ceqq \frac{1}{l^D}
 S\!\left((\mathcal T_L(\rho_L))_{Q_{L,l}(0)}\right).
 \label{eq:supp-twirled-entropy}
\end{align}
All constants implicit in $O(\cdot)$ are independent of
$L,l,\epsilon$, and $\rho_L$.
\end{prop}

\begin{proof}
Fix an offset $a$ and write $\Acal=\Acal_{L,a}^{(l)}$ and
$R=\Lambda_L\setminus\bigcup\Acal$.  Let
\begin{equation}
 H^R_{L,\Acal}\ceqq H_L-\sum_{A\in\Acal}H_{L,A}.
\end{equation}
The boundary estimate proved in
Sec.~\ref{sec:supp-boundary-residual-interactions} gives
\begin{equation}
 \|H^R_{L,\Acal}\|
 \le C_{D,k_\Phi}E_0\bigl[\bar{r}_F(l)+n_l(L)\bigr]L^D.
 \label{eq:supp-residual-used}
\end{equation}
For any unitary $U$, the minimum-energy characterization in
Eq.~\eqref{eq:supp-canonical-energy} yields
\begin{align}
 \Tr(H_LU\rho_LU^\dagger)
 &\ge
 \sum_{A\in\Acal}|A|
 u_{L,A}(s_A(U\rho_LU^\dagger))
 -\|H^R_{L,\Acal}\|.
 \label{eq:supp-min-energy-step}
\end{align}
The small incremental entangling theorem~\cite{bravyiUpperBoundsEntangling2007,vanacoleyenEntanglementRatesArea2013} and
Eq.~\eqref{eq:supp-compatibility} give
\begin{equation}
 \abs{s_A(\rho_L)-s_A(U\rho_LU^\dagger)}
 \le c_{\mathrm{SIE}}\ln(d)\,\epsilon.
 \label{eq:supp-entropy-loss-used}
\end{equation}
Proposition~\ref{prop:supp-uniform-holder}, proved below by adapting the
high-temperature cluster expansion of
Refs.~\cite{kimThermalAreaLaw2025,kimErratumThermalArea2025}, gives,
uniformly in $L$ and $A$,
\begin{equation}
 \abs{u_{L,A}(s)-u_{L,A}(s')}\le C E_0\sqrt{\abs{s-s'}}
 \qquad(0\le s,s'\le\ln d).
 \label{eq:supp-continuity-used}
\end{equation}
Furthermore, for translated cubes,
a standard estimate yields
\begin{equation}
 \sup_{s\in[0,\ln d]}
 |u_{L,Q_{L,l}(z)}(s)-u(s)|
 \le C_{D,k_\Phi}E_0\bar{r}_F(l).
 \label{eq:supp-u-finite-size-used}
\end{equation}
See Sec.~\ref{sec:supp-technical-estimates} for the derivation of these inequalities.
Combining Eqs.~\eqref{eq:supp-initial-energy} 
and~\eqref{eq:supp-residual-used}--\eqref{eq:supp-u-finite-size-used}, 
we obtain
\begin{align}
 w_L(\rho_L;\Ucal_L(\epsilon;l))
 \le
 u^{\eq}
 -
 \frac{1}{L^D}\sum_{A\in\Acal}|A|
 u(s_A(\rho_L))+E_0O\!\left(\sqrt\epsilon+\bar{r}_F(l)
 +n_l(L)+L^{-1/2}\right).
 \label{eq:supp-fixed-partition-bound}
\end{align}
Since $\abs{u(s)}\le E_0$ and
$n_l(L)=O(L^{-1/2})$, this proves
Eq.~\eqref{eq:supp-single-level-local}.

For a homogeneous protocol, both $H_L$ and $U$ commute with translations.
Consequently,
\begin{equation}
 w_L(\rho_L;\Ucal_L^{\mathrm{hom}}(\epsilon;l))
 =w_L(\mathcal T_L(\rho_L);\Ucal_L^{\mathrm{hom}}(\epsilon;l))
 \le w_L(\mathcal T_L(\rho_L);\Ucal_L(\epsilon;l)).
\end{equation}
All translated $l$-cubes of $\mathcal T_L(\rho_L)$ have entropy density
$\bar{s}_l(\rho_L)$.  Applying
Eq.~\eqref{eq:supp-single-level-local} to the twirled state proves
Eq.~\eqref{eq:supp-single-level-hom}.
\end{proof}

Averaging over $a\in\Lambda_L$ gives
\begin{align}
  w_L(\rho_L;\Ucal_L(\epsilon;l))
 &\le
 \bar{w}_L(\rho_L;l)
 +E_0O\!\left(\sqrt\epsilon+\bar{r}_F(l)+L^{-1/2}\right)\\
\bar{w}_L(\rho_L;l)&\ceqq u^\eq-\ave{u(s_{Q_{L,l}(z)}(\rho_L))}{z\in\Lambda_L}.
\end{align}
If $\rho_L$ is in MITE on scale $l(\gg1)$, 
$s_{Q_{L,l}(z)}(\rho_L)\simeq s^\eq$ and no extensive work can be extracted by 
any $U\in\Ucal_L(\epsilon;l)$ beyond the density $E_0O(\sqrt{\epsilon})$.
In particular, by short-range operations for short time on scale $l$,
MITE states have zero extractable work density,
which recovers the result of Ref.~\cite{hokkyoUniversalUpperBound2025}.

Similarly to $\bar{w}_L(\rho_L;l)$, we define
\begin{equation}
\bar{w}_L^{\mathrm{hom}}(\rho_L;l)\ceqq u^\eq-u(\bar{s}_l(\rho_L)),
\end{equation}
which appears on the right-hand side of Eq.~\eqref{eq:supp-single-level-hom}.
If $\rho_L$ is in iMATE, $\bar{s}_l(\rho_L)= s^\eq+o(l^0)$~\cite{chibaSecondLawThermodynamics2026},
and no extensive work can be extracted by 
any $U\in\Ucal^{\mathrm{hom}}_L(\epsilon;l)$ for sufficiently large $l$, 
in particular, by homogeneous macroscopic operations for short time on scale $l$.
This corresponds to the macroscopic passivity established in Ref.~\cite{chibaSecondLawThermodynamics2026}.

The same proof applies to a fixed coarse block family.  In particular, if
$\widehat{\Acal}_{L,a}^{(l)}$ is a coarsening of $\Acal_{L,a}^{(l)}$, subadditivity of the
boundary strength under unions gives
\begin{equation}
 w_L(\rho_L;\Ucal_L(\epsilon;\widehat{\Acal}_{L,a}^{(l)}))
 \le \bar{w}_L(\rho_L;\widehat{\Acal}_{L,a}^{(l)})
 +E_0O\!\left(\sqrt\epsilon+\bar{r}_F(l)+L^{-1/2}\right).
 \label{eq:supp-coarse-single-level}
\end{equation}

\subsection{Work extraction between thermodynamic levels}

For later use, define the spatial mean and fluctuation of the entropy density
at scale $l$ by
\begin{align}
 s_l(\rho_L)
 &\ceqq\frac{1}{L^D}\sum_{z\in\Lambda_L}s_{Q_{L,l}(z)}(\rho_L),\\
 \delta s_l(\rho_L)
 &\ceqq
 \left[\frac{1}{L^D}\sum_{z\in\Lambda_L}
 \bigl(s_{Q_{L,l}(z)}(\rho_L)-s_l(\rho_L)\bigr)^2\right]^{1/2}.
 \label{eq:supp-entropy-fluctuation}
\end{align}
As assumed in the main text, we are mainly interested in states satisfying
\begin{equation}
 \delta s_l(\rho_L)\ll1,
 \label{eq:supp-delta-assumption}
\end{equation}
which covers translation-invariant states, 
pure product states, and those obtained from these states via $o(l)$-depth local quantum circuits.

\subsubsection{Increasing the spatial resolution of control}
The Holevo information density associated with the position of the local
subsystem is
\begin{equation}
 \tilde\chi_l(\rho_L)
 \ceqq\bar{s}_l(\rho_L)-s_l(\rho_L)\ge0.
 \label{eq:supp-Holevo-density}
\end{equation}

\begin{thm}[Hierarchical work bound: spatial resolution]
\label{thm:supp-spatial-resolution}
For every $1\le l\le L/2$, we have
\begin{align}
\bar{w}_L(\rho_L;l)
&=\bar{w}_L^{\mathrm{hom}}(\rho_L;l)+u(\bar{s}_l(\rho_L))-\ave{u(s_{Q_{L,l}(z)}(\rho_L))}{z\in\Lambda_L}\nonumber\\
&\le \bar{w}_L^{\mathrm{hom}}(\rho_L;l)
 +\int_{\bar{s}_l(\rho_L)-\tilde\chi_l(\rho_L)}^{\bar{s}_l(\rho_L)}T(s)\,\d s.
 \label{eq:supp-spatial-resolution}
\end{align}
The gap between the two sides is nonnegative and bounded by $E_0O\!\left(\sqrt{\delta s_l(\rho_L)}\right)$.
\end{thm}
\begin{proof}
The first equality follows from the definitions.
The convexity of $u(s)$ yields the inequality.

By Proposition~\ref{prop:supp-uniform-holder} and 
Eq.~\eqref{eq:supp-torus-cube-u-limit}, we have
\begin{equation}
  \sup_{s\ne s'\in[0,\ln d]}\frac{\abs{u(s)-u(s')}}{\sqrt{\abs{s-s'}}}\le E_0 c
\end{equation}
for some $c<\infty$; 
see Sec.~\ref{sec:supp-uniform-holder} and~\ref{sec:supp-u-limit} for the proof.
Using this, we have
\begin{align}
  \ave{u(s_{Q_{L,l}(z)}(\rho_L))}{z\in\Lambda_L}-u(s_l(\rho_L))
  &\le \frac{1}{L^D}\sum_{z\in\Lambda_L;s_{Q_{L,l}(z)}(\rho_L)\ge s_l(\rho_L)}\abs{u(s_{Q_{L,l}(z)}(\rho_L))-u(s_l(\rho_L))}\nonumber\\
  &\le E_0c\frac{1}{L^D}\sum_{z\in\Lambda_L;s_{Q_{L,l}(z)}(\rho_L)\ge s_l(\rho_L)}\abs{s_{Q_{L,l}(z)}(\rho_L)-s_l(\rho_L)}^{1/2}\nonumber\\
  &\le E_0c\ave{\abs{s_{Q_{L,l}(z)}(\rho_L)-s_l(\rho_L)}^{1/2}}{z\in\Lambda_L}\nonumber\\
  &\le E_0c\sqrt{\delta s_l(\rho_L)},
  \label{eq:supp-inverse-jensen}
\end{align}
which completes the proof.
\end{proof}

\subsubsection{Extending the control duration}
Let $1\le l<\hat{l}\le L/2$, 
where $\hat{l}$ is an integer multiple of $l$. 
We are interested in the regime $\hat{l}\gg l$. 
For each $w\in\Lambda_L$, 
let $\mathcal Z_{l|\hat{l}}(w)$ be the unique set of translation parameters such that the cubes $\{Q_{L,l}(z)\mid z\in\mathcal Z_{l|\hat{l}}(w)\}$ form a partition of $Q_{L,\hat{l}}(w)$. 
We define 
\begin{align}
  s_{l,w}(\rho_L)&=\ave{s_{Q_{L,l}(z)}(\rho_L)}{z\in\mathcal Z_{l|\hat{l}}(w)},\\
  i_w(l\to\hat{l})(\rho_L)&=s_{l,w}(\rho_L)-s_{Q_{L,\hat{l}}(w)}(\rho_L)=\frac{1}{\hat{l}^D}\left[
    S\left(\bigotimes_{z\in\mathcal Z_{l|\hat{l}}(w)}(\rho_L)_{Q_{L,l}(z)}\right)-S\left((\rho_L)_{Q_{L,\hat{l}}(w)}\right)
  \right].
\end{align}
The quantity $i_w(l\to\hat{l})(\rho_L)$ is the total correlation per site among the $l$-cubes in this partition. 
We omit the subscript $w$ when the quantity is independent of $w$.
Then we have the following.
\begin{thm}[Hierarchical work bound: operation time]\label{thm:supp-operation-time}
  \begin{equation}
    \bar{w}_L(\rho_L;\hat{l})
    =
    \bar{w}_L(\rho_L;l)
    +
    \ave{
    \int_{s_{l,w}(\rho_L)-i_w(l\to\hat{l})(\rho_L)}^{s_{l,w}(\rho_L)}T(s)\d s
    }{w\in\Lambda_L}
    +E_0O(\sqrt{\delta s_l(\rho_L)}).
    \label{eq:supp-short-to-long}
  \end{equation}
In particular, 
suppose that every translation changes the state $\rho_L$ only by a product of on-site unitaries. 
Then $s_{l,w}(\rho_L)$ and $i_w(l\to\hat{l})(\rho_L)$ are independent of $w$, and the relation reduces to 
\begin{equation}
  \bar{w}_L(\rho_L;\hat{l})
    =
    \bar{w}_L(\rho_L;l)
    +
    \int_{s_{l}(\rho_L)-i(l\to\hat{l})(\rho_L)}^{s_{l}(\rho_L)}T(s)\d s.
\end{equation}
\end{thm}
\begin{proof}
By definition, we have
\begin{align}
\bar{w}_L(\rho_L;\hat{l})
-\bar{w}_L(\rho_L;l)
&=\ave{u(s_{Q_{L,l}(z)}(\rho_L))}{z\in\Lambda_L}-\ave{u(s_{Q_{L,\hat{l}}(w)}(\rho_L))}{w\in\Lambda_L}.
\end{align}
For each $a\in\Lambda_L$, consider the large-block family $\Acal_{L,a'}^{(\hat{l})}$ 
obtained by grouping the blocks of $\Acal_{L,a}^{(l)}$ into cubes of side length $\hat{l}$, 
where 
\begin{equation}
 a'\ceqq
 a+_L\left(\left\lceil\frac{\hat{l}}{2}\right\rceil
 -\left\lceil\frac{l}{2}\right\rceil\right)\mathbf 1,
 \qquad \mathbf 1=(1,\ldots,1).
\end{equation}
Define $\Acal_{L,a}^{(l)}(\hat{A})\ceqq\{A\in\Acal_{L,a}^{(l)}\mid A\subseteq\hat{A}\}$ 
and $\Acal_{L,a}^{\circ}\ceqq \bigcup_{\hat A\in\mathcal A_{L,a'}^{(\hat l)}}\Acal_{L,a}^{(l)}(\hat{A})$.
Then we have
\begin{align}
  \ave{u(s_{Q_{L,l}(z)}(\rho_L))}{z\in\Lambda_L}
  &=\ave{\ave{u(s_A(\rho_L))}{A\in\Acal_{L,a}^\circ}}{a\in\Lambda_L}\nonumber\\
  &=\ave{\ave{\ave{u(s_A(\rho_L))}{A\in\Acal_{L,a}^{(l)}(\hat{A})}}{\hat{A}\in\Acal_{L,a'}^{(\hat{l})}}}{a\in\Lambda_L}\nonumber\\
  &=\ave{\ave{u(s_{l,\hat{A}}(\rho_L))}{\hat{A}\in\Acal_{L,a'}^{(\hat{l})}}}{a\in\Lambda_L}
  +\ave{\ave{\ave{u(s_A(\rho_L))}{A\in\Acal_{L,a}^{(l)}(\hat{A})}-u(s_{l,\hat{A}}(\rho_L))
  }{\hat{A}\in\Acal_{L,a'}^{(\hat{l})}}}{a\in\Lambda_L}\nonumber\\
  &=\ave{u(s_{l,w}(\rho_L))}{w\in\Lambda_L}
  +\ave{\ave{\ave{u(s_A(\rho_L))}{A\in\Acal_{L,a}^{(l)}(\hat{A})}-u(s_{l,\hat{A}}(\rho_L))
  }{\hat{A}\in\Acal_{L,a'}^{(\hat{l})}}}{a\in\Lambda_L}.
  \label{eq:supp-short-to-long-1}
\end{align}
Here, $s_{l,\hat{A}}(\rho_L)\ceqq\ave{s_{A}(\rho_L)}{A\in\Acal_{L,a}^{(l)}(\hat{A})}$ 
is the average of $s_A(\rho_L)$ over the small blocks $A\subseteq\hat{A}$.

The difference between $\ave{u(s_{l,w}(\rho_L))}{w\in\Lambda_L}$ and $\ave{u(s_{Q_{L,\hat{l}}(w)}(\rho_L))}{w\in\Lambda_L}$ 
is the integral term in Eq.~\eqref{eq:supp-short-to-long}.
It therefore remains to bound the second term in the last line of Eq.~\eqref{eq:supp-short-to-long-1}.
Proceeding as in Eq.~\eqref{eq:supp-inverse-jensen}, we obtain 
\begin{align}
0\le\ave{\ave{u(s_A(\rho_L))}{A\in\Acal_{L,a}^{(l)}(\hat{A})}
-
u(s_{l,\hat{A}}(\rho_L))}{\hat{A}\in\Acal_{L,a'}^{(\hat{l})}}
&\le 
E_0c
\ave{\ave{\abs{s_A(\rho_L)-s_{l,\hat{A}}(\rho_L)}^{1/2}}{A\in\Acal_{L,a}^{(l)}(\hat{A})}}{\hat{A}\in\Acal_{L,a'}^{(\hat{l})}}\nonumber\\
&\le
E_0c
\ave{\ave{\abs{s_A(\rho_L)-s_{l,\hat{A}}(\rho_L)}^{2}}{A\in\Acal_{L,a}^{(l)}(\hat{A})}}{\hat{A}\in\Acal_{L,a'}^{(\hat{l})}}^{1/4}\nonumber\\
&=E_0c\left(
  \ave{s_A(\rho_L)^2}{A\in\Acal_{L,a}^{\circ}}-\ave{s_{l,\hat{A}}(\rho_L)^2}{\hat{A}\in\Acal_{L,a'}^{(\hat{l})}}
\right)^{1/4}\nonumber\\
&\le
E_0c\left(
  \ave{s_A(\rho_L)^2}{A\in\Acal_{L,a}^{\circ}}-\ave{s_{A}(\rho_L)}{A\in\Acal_{L,a}^{\circ}}^2
\right)^{1/4}.
\end{align}
Averaging over $a\in\Lambda_L$ gives
\begin{equation}
  \ave{\left(
  \ave{s_A(\rho_L)^2}{A\in\Acal_{L,a}^{\circ}}-\ave{s_{A}(\rho_L)}{A\in\Acal_{L,a}^{\circ}}^2
\right)^{1/4}}{a\in\Lambda_L}
\le \left(
  \ave{s_{Q_{L,l}(z)}(\rho_L)^2}{z\in\Lambda_L}-\ave{\ave{s_{A}(\rho_L)}{A\in\Acal_{L,a}^{\circ}}^2}{a\in\Lambda_L}
\right)^{1/4}
\le \sqrt{\delta s_l(\rho_L)}
\end{equation}
and hence the remaining term is bounded by $CE_0\sqrt{\delta s_l(\rho_L)}$. 
This proves the theorem.
\end{proof}

For short-range operations, this comparison becomes relevant when the short-time condition at scale $l$ is no longer satisfied. 
Suppose that $E_0'T_L=\Omega(\bar{r}_{F'}(l)^{-1})$, while the protocol still satisfies $E_0'T_L=o(\bar{r}_{F'}(\hat{l})^{-1})$. 
In this regime, Eq.~\eqref{eq:supp-short-to-long} gives an upper bound containing an additional term determined by the correlations among the $l$-scale subsystems. 
These correlations are not accessible under the short-time condition at scale $l$.

\subsubsection{Relaxing the locality constraint on control}
We fix an offset $a\in\Lambda_L$ and omit it from the notation for simplicity.
Let
$\widehat{\Acal}_L^{(l)}$ be an exact coarsening of $\Acal_L^{(l)}$, so that
every $\hat{A}\in\widehat{\Acal}_L^{(l)}$ is a disjoint union of the blocks in
\begin{equation}
 \Acal_L^{(l)}(\hat{A})
 \ceqq\{A\in\Acal_L^{(l)}:A\subseteq\hat{A}\}.
\end{equation}
For each $\hat{A}$, define
\begin{align}
 s_{l,\hat{A}}(\rho_L)
 &\ceqq
 \ave{s_A(\rho_L)}{A\in\Acal_L^{(l)}(\hat{A})},\\
 i_{\hat{A}}(\rho_L)
 &\ceqq
 s_{l,\hat{A}}(\rho_L)-s_{\hat{A}}(\rho_L)\nonumber\\
 &=
 \frac{1}{|\hat{A}|}
 \left[
 \sum_{A\in\Acal_L^{(l)}(\hat{A})}S((\rho_L)_A)
 -S((\rho_L)_{\hat{A}})
 \right]\ge0.
 \label{eq:supp-coarse-correlation}
\end{align}
The quantity $i_{\hat{A}}$ is the total correlation density among the small
blocks contained in $\hat{A}$.
\begin{thm}[Hierarchical work bound: control locality]
\label{thm:supp-control-locality}
  \begin{equation}
    \bar{w}_L(\rho_L;\widehat{\Acal}_{L}^{(l)})
  =
  \bar{w}_L(\rho_L;\Acal_{L}^{(l)})
  +
  \ave{\int_{s_{l,\hat{A}}(\rho_L)-i_{\hat{A}}(\rho_L)}^{s_{l,\hat{A}}(\rho_L)}T(s)\d s}
  {{\hat{A}}\in\widehat{\Acal}_L^{(l)}}
  +E_0O((\delta s_l(\rho_L))^{1/(D+2)}).
  \end{equation}
\end{thm}
\begin{proof}
  \begin{align}
    &\bar{w}_L(\rho_L;\widehat{\Acal}_{L}^{(l)})
    -\bar{w}_L(\rho_L;\Acal_{L}^{(l)})\nonumber\\
    &=\ave{\ave{u(s_A(\rho_L))}{A\in\Acal_{L}^{(l)}(\hat{A})}-u(s_{\hat{A}}(\rho_L))}{{\hat{A}}\in\widehat{\Acal}_L^{(l)}}\nonumber\\
    &=\ave{\int_{s_{l,\hat{A}}(\rho_L)-i_{\hat{A}}(\rho_L)}^{s_{l,\hat{A}}(\rho_L)}T(s)\d s}
  {{\hat{A}}\in\widehat{\Acal}_L^{(l)}}
  +\ave{\ave{u(s_A(\rho_L))}{A\in\Acal_{L}^{(l)}(\hat{A})}-u(s_{l,\hat{A}}(\rho_L))}{{\hat{A}}\in\widehat{\Acal}_L^{(l)}}.
  \end{align}
  The second term in the last line is bounded 
  by the variance of $s_A$ over $A\in\Acal_L^{(l)}(\hat{A})$, as in Eq.~\eqref{eq:supp-inverse-jensen}.
  Using the continuity bound, $\abs{s_A-s_{A+_Lx}}\le \abs{A\triangle (A+_Lx)}/\abs{A}\ln d$, 
  we can further bound this variance by the spatial variance $\delta s_l(\rho_L)$ as follows.
  
  Let
  $f(z)=s_{Q_{L,l}(z)}(\rho_L)$,
  $\mu=s_l(\rho_L)$, and $\delta=\delta s_l(\rho_L)$.  The conditional
  variance satisfies
  \begin{align}
  V_{\Acal_L^{(l)}|\widehat{\Acal}_L^{(l)}}
  &\ceqq
  \ave{
  \ave{[s_A(\rho_L)-s_{l,\hat{A}}(\rho_L)]^2}
        {A\in\Acal_L^{(l)}(\hat{A})}
  }{\hat{A}\in\widehat{\Acal}_L^{(l)}}\nonumber\\
  &\le
  \ave{[s_A(\rho_L)-\mu]^2}{A\in\Acal_L^{(l)}}.
  \label{eq:supp-conditional-variance}
  \end{align}
  To estimate the last term from the spatial variance $\delta^2$, let
  $q=\lfloor L/l\rfloor$ and let
  $Z_a=\{a+_L(ln):n\in\{0,\ldots,q-1\}^D\}$ be the set of small-block
  centers.  For $1\le r\le l$, set
  $K_r=\{0,\ldots,r-1\}^D$.  The points $z+_Lx$ with
  $z\in Z_a$ and $x\in K_r$ are all distinct.  
  Moreover, the Araki--Lieb inequality~\cite{arakiEntropyInequalities1970} gives
  \begin{equation}
  |f(z+_Lx)-f(z)|
  \le
  \frac{|Q_{L,l}(z+_Lx)\mathbin{\triangle}Q_{L,l}(z)|}{l^D}\ln d
  \le C_D\frac{r-1}{l}\ln d .
  \label{eq:supp-entropy-translation}
  \end{equation}
  Averaging
  $|f(z)-\mu|^2\le
  2|f(z)-f(z+_Lx)|^2+2|f(z+_Lx)-\mu|^2$
  over $z\in Z_a$ and $x\in K_r$ therefore yields, 
  \begin{equation}
  \ave{[s_A(\rho_L)-\mu]^2}{A\in\Acal_L^{(l)}}
  \le
  2(C_D\ln d)^2\left(\frac{r-1}{l}\right)^2
  +2\left(\frac{L}{qr}\right)^D\delta^2
  \le C_{D}\left[
  \left(\frac{r-1}{l}\right)^2(\ln d)^2
  +\left(\frac{l}{r}\right)^D\delta^2
  \right].
  \label{eq:supp-sampled-variance}
  \end{equation}
  Here we used $l\le L/2$, 
  which implies $L/q\le (l^{-1}-L^{-1})^{-1}\le2l$. 
  If $\delta>0$, we choose $r\ceqq \lceil l(\delta/\ln d)^{2/(D+2)} \rceil$ in Eq.~\eqref{eq:supp-sampled-variance}.  This gives
  \begin{equation}
  \ave{[s_A(\rho_L)-\mu]^2}{A\in\Acal_L^{(l)}}
  \le C_{D,d}\delta^{4/(D+2)}.
  \label{eq:supp-sampled-variance-final}
  \end{equation}
  Thus Eq.~\eqref{eq:supp-sampled-variance-final} holds if $\delta>0$.
  Note that Eq.~\eqref{eq:supp-sampled-variance-final} trivially holds if $\delta=0$.
  The H\"older continuity of $u$ completes the proof.
\end{proof}
Reference~\cite{hokkyoUniversalUpperBound2025} demonstrated with an explicit toy model that long-range operations can extract energy even from states that are thermal on spatial scales of order $L$. 
The theorem above provides a general quantitative bound on this effect.

\subsection{Technical estimates}
\label{sec:supp-technical-estimates}

\subsubsection{Boundary and residual interactions}
\label{sec:supp-boundary-residual-interactions}
The periodized interaction in Eq.~\eqref{eq:supp-periodized-interaction} 
satisfies the finite-volume analogue of Eq.~\eqref{eq:supp-Phi-locality}:
\begin{equation}
 \sum_{Y\Subset\Lambda_L,Y\ni x,y}\frac{\|\Phi_L(Y)\|}{|Y|}
 \le E_0F(d_L(x,y))
 \label{eq:supp-finite-pair-strength}
\end{equation}
for every $x,y\in\Lambda_L$.  

To prove this, fix $x,y\in\Lambda_L$ and set
$z\ceqq y-_Lx\in\Lambda_L$.  
If one coordinate of $z$ has absolute value $L/2$, 
no $L$-small set can contain both $x$ and $y$, 
and the left-hand side of Eq.~\eqref{eq:supp-finite-pair-strength} vanishes.  
We may therefore assume that every coordinate of $z$ has
absolute value strictly smaller than $L/2$.
Let $Y$ be an $L$-small set containing $x$ and $y$, let $\hat{X}$ be its lift, 
and denote the preimages of $x$ and $y$ by $\hat{x}$ and
$\hat{y}$.  Set $X\ceqq\hat{X}-\hat{x}$.  
Then $0\in X,\diam X<L/2$ and $\hat{y}-\hat{x}=z\in X$.  
Translation invariance and the definition of $\Phi_L$ give
\begin{equation}
 |Y|=|X|,
 \qquad
 \|\Phi_L(Y)\|=\|\Phi(\hat{X})\|=\|\Phi(X)\|.
 \label{eq:supp-finite-pair-lift-norm}
\end{equation}
Conversely, every
\begin{equation}
 X\in\mathfrak S_{L,z}
 \ceqq
 \left\{X\Subset\mathbb Z^D:
 0,z\in X,\ \diam X<L/2\right\}
\end{equation}
determines the $L$-small set $Y=x+_LX$.  
The uniqueness of the lift shows that these two constructions are inverse to each other.  Consequently,
\begin{align}
 \sum_{Y\Subset\Lambda_L,Y\ni x,y}\frac{\|\Phi_L(Y)\|}{|Y|}
 &=\sum_{X\in\mathfrak S_{L,z}}
 \frac{\|\Phi(X)\|}{|X|}\nonumber\\
 &\le
 \sum_{\substack{X\Subset\Z^D,\\X\ni0,z}}
 \frac{\|\Phi(X)\|}{|X|}\nonumber\\
 &\le E_0F(|z|)
 =E_0F(d_L(x,y)),
\end{align}
where the last inequality is Eq.~\eqref{eq:supp-Phi-locality}, 
and the last equality follows from Eq.~\eqref{eq:supp-periodic-distance}.

For $A\subseteq\Lambda_L$, define its boundary strength by
\begin{equation}
 b_{\Phi,L}(A)
 \ceqq\sum_{Y\Subset\Lambda_L;Y\cap A\ne\emptyset,Y\cap A^c\ne\emptyset}\|\Phi_L(Y)\|.
 \label{eq:supp-boundary-strength}
\end{equation}
For a support $Y$ that crosses the cut between $A$ and $A^c$, 
the same argument as Eq.~\eqref{eq:supp-crossing-multiplicity} yields
\begin{align}
 b_{\Phi,L}(A)
 &\le2\sum_{x\in A}\sum_{y\in \Lambda_L\setminus A}
 \sum_{Y\Subset\Lambda_L,Y\ni x,y}\frac{\|\Phi_L(Y)\|}{|Y|}
 \nonumber\\
 &\le2E_0\sum_{x\in A}\sum_{y\in \Lambda_L\setminus A}F(d_L(x,y)).
 \label{eq:supp-boundary-pair-bound}
\end{align}
For a set $Q\Subset\Lambda_L$, let $m_Q(x)\ceqq\min_{y\in\Lambda_L\setminus Q}d_L(x,y)-1$ be the distance between a site
$x\in Q$ and $\Lambda_L\setminus Q$.  
Then
\begin{equation}
 \sum_{y\in\Lambda_L\setminus Q}F(d_L(x,y))
 \le r_F(m_Q(x)).
\end{equation}
For each cube $Q=Q_{L,l}(z)$ and $m$, the corresponding layer 
$\{x\in Q\mid m\le m_Q(x)<m+1\}$
contains at most
$\kappa_{D}l^{D-1}$ sites, and there are fewer than $l$ nonempty layers.
Therefore,
\begin{align}
 \sum_{x\in Q_{L,l}(z)}
 \sum_{y\in \Lambda_L\setminus Q_{L,l}(z)}F(d_L(x,y))
 &\le \kappa_{D}l^{D-1}\sum_{m=0}^{l-1}r_F(m)\nonumber\\
 &\le \kappa_{D}l^D\bar{r}_F(l).
 \label{eq:supp-cube-boundary}
\end{align}

More generally, let $\Acal$ be any disjoint family and
$R=\Lambda_L\setminus\bigcup\Acal$.  Every term in
$H^R_{L,\Acal}$ either crosses the boundary of at least one $A\in\Acal$ or
intersects $R$.  Hence
\begin{align}
 \|H^R_{L,\Acal}\|
 &\le\sum_{A\in\Acal}b_{\Phi,L}(A)
 +\sum_{Y:Y\cap R\ne\varnothing}\|\Phi_L(Y)\|
 \nonumber\\
 &\le\sum_{A\in\Acal}b_{\Phi,L}(A)
 +k_\Phi E_0F(0)|R|.
 \label{eq:supp-general-residual}
\end{align}
The second inequality follows from
$\sum_{Y\ni x}\|\Phi_L(Y)\|\le k_\Phi E_0F(0)$, which is
Eq.~\eqref{eq:supp-finite-pair-strength} with $x=y$, multiplied by the
finite-body bound.
Equations~\eqref{eq:supp-boundary-pair-bound}, 
\eqref{eq:supp-cube-boundary} and
\eqref{eq:supp-general-residual} imply
Eq.~\eqref{eq:supp-residual-used}.  They also show that
$b_{\Phi,L}(A\cup B)\le b_{\Phi,L}(A)+b_{\Phi,L}(B)$ for disjoint $A,B$,
which is the bound used for Eq.~\eqref{eq:supp-coarse-single-level}.

\subsubsection{Entropy change under a compatible control}

For a term $\Psi_{L,t}(X)$ crossing the cut between $A$ and $A^c$, 
the Hilbert spaces of its support on the two sides have dimensions $d^{|X\cap A|}$ and
$d^{|X\cap A^c|}$.  Applying the small-incremental-entangling theorem~\cite{bravyiUpperBoundsEntangling2007,vanacoleyenEntanglementRatesArea2013,audenaertQuantumSkewDivergence2014} to each
term and using the linearity of the entropy derivative in the Hamiltonian
gives~\cite{marienEntanglementRatesStability2016,gongEntanglementAreaLaws2017}
\begin{equation}
 \left|\frac{\d}{\d t}S_A(\rho_L(t))\right|
 \le c_{\mathrm{SIE}}\ln d
 \sum_{X\Subset \Lambda_L}\mu_A(X)\|\Psi_{L,t}(X)\|,
 \label{eq:supp-SIE-rate}
\end{equation}
where $c_{\mathrm{SIE}}\le 6.5$ is a numerical constant.
Whenever $\mathfrak B_{\Psi,L}(A)/|A|\le\epsilon$, integration yields
\begin{equation}
 |s_A(\rho_L(T_L))-s_A(\rho_L(0))|
 \le c_{\mathrm{SIE}}\ln(d)\,\epsilon,
 \label{eq:supp-SIE-integrated}
\end{equation}
which proves Eq.~\eqref{eq:supp-entropy-loss-used}.

\subsubsection{High-temperature clustering and the energy-variance bound}
\label{sec:supp-high-temperature-variance}

To obtain the continuity bound for the canonical energy used in the above proofs, 
we show a uniform bound on the energy variance. 
This is obtained from a high-temperature clustering estimate, 
which is a slight modification of the clustering estimate obtained in Ref.~\cite{kimThermalAreaLaw2025,kimErratumThermalArea2025}.

For a self-adjoint operator $K$, we define
\begin{align}
 \rho_{\beta,K}
 &\ceqq\frac{e^{-\beta K}}{\Tr e^{-\beta K}},\\
 \operatorname{Cor}_{\beta,K}(O,O')
 &\ceqq
 \Tr(\rho_{\beta,K}OO')
 -\Tr(\rho_{\beta,K}O)\Tr(\rho_{\beta,K}O'),\\
 \operatorname{Var}_{\beta,K}(O)
 &\ceqq\operatorname{Cor}_{\beta,K}(O,O)
 =\Tr(\rho_{\beta,K}O^2)-\Tr(\rho_{\beta,K}O)^2
 \label{eq:supp-correlation-variance}
\end{align}
for self-adjoint $O$ in the last line.

\begin{lem}[High-temperature clustering for a summable profile]
\label{lem:supp-high-temperature-clustering}
Let $F$ be nonincreasing and satisfy
Eq.~\eqref{eq:supp-F-normalization}.  There exist a nonincreasing function
$G:[0,\infty)\to(0,\infty)$ with $F\le G$, constants $C_G<\infty$ and
$\beta_{\mathrm{c}}>0$, and, for every $\beta_0<\beta_{\mathrm{c}}$, a constant
$C_{\beta_0}<\infty$ such that
\begin{equation}
 \|G\|_1\ceqq\sum_{z\in\Z ^D}G(|z|)<\infty,
 \label{eq:supp-G-summable}
\end{equation}
and
\begin{equation}
 \sum_{z\in\Z ^D}G(|x-z|)G(|z-y|)
 \le C_GG(|x-y|)
 \label{eq:supp-G-convolution}
\end{equation}
and, for every nonempty $A\subseteq\Lambda_L$, $0\le\beta\le\beta_0$, and
operators $O_X,O_Y$ supported on $X,Y\subseteq A$,
\begin{align}
 |\operatorname{Cor}_{\beta,H_{L,A}}(O_X,O_Y)|
 &\le C_{\beta_0}\|O_X\|\|O_Y\|
 e^{(|X|+|Y|)/k_\Phi}
 \sum_{x\in X}\sum_{y\in Y}\mathcal G_L(x,y),
 \label{eq:supp-high-temperature-clustering}
\end{align}
where $\mathcal G_L:\Lambda_L\times\Lambda_L\to(0,\infty)$ is defined by
\begin{equation}
 \mathcal G_L(x,y)
 \ceqq\sum_{n\in\Z ^D}
 G(| x- y+Ln|).
 \label{eq:supp-periodized-G}
\end{equation}  
The constants are independent of $L$ and $A$.
\end{lem}

\begin{proof}
\emph{Construction of the majorant.}
For a nonincreasing radial function $G$, the condition
$\norm{G}_1<\infty$ is equivalent to
$\sum_{n\ge0}2^{Dn}G(2^n)<\infty$. 
This is because there exist constants
$0<c_{D}\le C_{D}<\infty$ such that, for all $n\in\Z_{\ge 1}$,
\begin{equation}
  c_{D}2^{Dn}\le
  \abs{\{z\in\Z^D:2^{n-1}\le |z|<2^{n}\}}
  \le C_{D}2^{Dn}.
  \label{eq:supp-shell-volume}
\end{equation}

Fix $q>2^D$ and set
\begin{align}
 f_{-1}&=F(0),& f_n&=F(2^n)\quad(n\ge0),\\
 g_{-1}&=f_{-1},&
 g_n&=\max\{f_n,q^{-1}g_{n-1}\}\quad(n\ge0).
\end{align}
Define $G(r)=g_{-1}$ for $0\le r<1$ and
$G(r)=g_n$ for $2^n\le r<2^{n+1}$.  Then $G$ is nonincreasing,
$F\le G$, and
\begin{equation}
 G(r/2)\le qG(r).
 \label{eq:supp-G-doubling}
\end{equation}
The shell-volume bound~\eqref{eq:supp-shell-volume} and Eq.~\eqref{eq:supp-F-normalization} imply
$\sum_{n\ge -1}2^{Dn}f_n<\infty$. 
Moreover,
$g_n\le f_n+q^{-1}g_{n-1}$ yields
\begin{equation}
 g_n\le\sum_{m=-1}^nq^{-(n-m)}f_m,
\end{equation}
and therefore
\begin{equation}
 \sum_{n\ge0}2^{Dn}g_n
 \le\frac{C_D}{1-2^D/q}\sum_{n\ge-1}2^{Dn}f_n<\infty.
\end{equation}
This proves Eq.~\eqref{eq:supp-G-summable}.  Finally, for
$R=|x-y|$, every $z$ satisfies $|x-z|\ge R/2$ or $|z-y|\ge R/2$.
Using Eq.~\eqref{eq:supp-G-doubling},
\begin{align}
 \sum_zG(|x-z|)G(|z-y|)
 &\le2G(R/2)\|G\|_1
 \le2q\|G\|_1G(R).
\end{align}
Thus Eq.~\eqref{eq:supp-G-convolution} holds, and we may take
$C_G=2q\|G\|_1$.

\emph{Periodization and the cluster expansion.}
Reindexing the periodic images as points of $\Z ^D$ gives
\begin{align}
 \sum_{y\in\Lambda_L}\mathcal G_L(x,y)&=\|G\|_1,
 \label{eq:supp-periodized-G-row-sum}\\
 \sum_{z\in\Lambda_L}\mathcal G_L(x,z)\mathcal G_L(z,y)
 &\le C_G\mathcal G_L(x,y).
 \label{eq:supp-periodized-G-convolution}
\end{align}

To apply the cluster expansion of
Kim, Kuwahara, and Saito~\cite{kimThermalAreaLaw2025,kimErratumThermalArea2025}, define the unweighted
pair strength
\begin{equation}
 J_L(x,y)\ceqq\sum_{Y\Subset\Lambda_L;Y\ni x,y}\|\Phi_L(Y)\|.
\end{equation}
Equations~\eqref{eq:supp-k-body},
\eqref{eq:supp-finite-pair-strength}, and $F\le G$ imply
\begin{equation}
 J_L(x,y)
 \le k_\Phi E_0F(d_L(x,y))
 \le k_\Phi E_0\mathcal G_L(x,y).
 \label{eq:supp-KKS-pair}
\end{equation}
The clustering theorem of Refs.~\cite{kimThermalAreaLaw2025,kimErratumThermalArea2025} 
assumes a power-law profile $F(r)\in O(r^{-\alpha})$ for $\alpha>D$,
but it is used only in two estimates: a uniform row-sum bound and the convolution estimate.
Equations~\eqref{eq:supp-periodized-G-row-sum} and~\eqref{eq:supp-periodized-G-convolution} provide these estimates
with constants $\|G\|_1$ and $C_G$, respectively.  
The remaining estimates in
Refs.~\cite{kimThermalAreaLaw2025,kimErratumThermalArea2025} use only the
row-sum bound, the convolution bound, and $|X|\le k_\Phi$.  Replacing the
power-law profile by $\mathcal G_L$ and retaining the sums over the two
endpoint sites, rather than bounding both endpoints by their minimum
separation, gives Eq.~\eqref{eq:supp-high-temperature-clustering}.  More explicitly, with
$\lambda=k_\Phi E_0$ and
$u_G=\max\{C_G,\|G\|_1\}$, the same counting proof converges whenever
\begin{equation}
 4\beta\lambda u_Gk_\Phi e^2
 e^{2\beta\lambda u_Gk_\Phi}<1,
 \label{eq:supp-high-temperature-condition}
\end{equation}
which defines a threshold $\beta_{\mathrm{c}}>0$.  
For every $\beta_0<\beta_{\mathrm{c}}$, 
the resulting constant in the counting argument is uniformly bounded on $[0,\beta_0]$, 
which gives $C_{\beta_0}<\infty$.  
This proves the claimed inequality.
\end{proof}

Note that Lemma~\ref{lem:supp-high-temperature-clustering} implies 
the algebraic decay of correlations,
although this form is not needed below.  
For $n\ge 0$, set
\begin{equation}
  h_n\ceqq
2^{-Dn}\sum_{k\ge n}2^{Dk}g_k.
\end{equation}
Any two points in
$(x-_L y)+L\Z^D$ 
are separated by at least $L$.  
Hence, if $2^n\le \abs{x-_Ly}<2^{n+1}$, 
a packing estimate shows that the number of 
periodic images in the shell $\{z\in(x-_L y)+L\Z^D\mid 2^k\le |z|<2^{k+1}\}$ is at most
$C_D2^{D(k-n)}$.  
Note that this is empty if $k<n$ by the definition of $x-_Ly$.
Therefore,
\begin{equation}
\mathcal G_L(x,y)
\le
C_D2^{-Dn}\sum_{k\ge n}2^{Dk}g_k=C_Dh_n.
\end{equation}
Moreover, $h_n$ is nonincreasing and
\begin{equation}
2^{Dn}h_n
=
\sum_{k\ge n}2^{Dk}g_k
\longrightarrow 0,
\end{equation}
because $\sum_n2^{Dn}g_n<\infty$.
Thus, by defining $\widetilde G$ to be a sufficiently large constant
on $[0,1)$ and $C_{D,q}h_n$ on each interval $[2^n,2^{n+1})$, we obtain
\begin{equation}
  \mathcal G_L(x,y)\le\widetilde G(d_L(x,y)),
\qquad
\widetilde G(r)=o(r^{-D}).
\end{equation}

Finally, we use the high-temperature clustering estimate to bound the canonical energy variance.  
\begin{cor}[Uniform high-temperature energy-variance bound]
\label{cor:supp-energy-variance}
For every $\beta_0<\beta_{\mathrm{c}}$, there is a constant $c_{\mathrm{var}}<\infty$
such that
\begin{equation}
 \sup_{\substack{L,\,\varnothing\ne A\subseteq\Lambda_L\\
                  0\le\beta\le\beta_0}}
 \frac{\operatorname{Var}_{\beta,H_{L,A}}(H_{L,A})}{|A|}
 \le c_{\mathrm{var}}E_0^2.
 \label{eq:supp-canonical-fluctuation}
\end{equation}
The constant is independent of $L$ and $A$.
\end{cor}

\begin{proof}
Applying Lemma~\ref{lem:supp-high-temperature-clustering} termwise to
$H_{L,A}=\sum_{X\Subset A}\Phi_L(X)$ gives the required variance bound.
Taking the absolute value of each connected-correlation term, we have
\begin{equation}
 \operatorname{Var}_{\beta,H_{L,A}}(H_{L,A})
 \le
 \sum_{X,Y\Subset A}
 \left|\operatorname{Cor}_{\beta,H_{L,A}}
 (\Phi_L(X),\Phi_L(Y))\right|.
\end{equation}
Set
\begin{equation}
 a_x^{(A)}\ceqq
 \sum_{\substack{X\Subset A\\X\ni x}}\|\Phi_L(X)\|.
\end{equation}
Equations~\eqref{eq:supp-k-body} and
\eqref{eq:supp-finite-pair-strength}, with $x=y$, imply
$a_x^{(A)}\le k_\Phi E_0F(0)$.  Since $|X|,|Y|\le k_\Phi$,
Eq.~\eqref{eq:supp-high-temperature-clustering} yields
\begin{align}
 \operatorname{Var}_{\beta,H_{L,A}}(H_{L,A})
 &\le C_{\beta_0} e^2
 \sum_{x,y\in A}\mathcal G_L(x,y)a_x^{(A)}a_y^{(A)}
 \nonumber\\
 &\le C_{\beta_0} e^2
 \bigl(k_\Phi E_0F(0)\bigr)^2\|G\|_1|A|.
 \label{eq:supp-variance-sum}
\end{align}
This proves Eq.~\eqref{eq:supp-canonical-fluctuation}.
\end{proof}

\subsubsection{Uniform H\"older continuity of the canonical energy}
\label{sec:supp-uniform-holder}

This subsection converts the energy-variance bound 
into the continuity estimate of the canonical energy.  
Near infinite temperature ($s=\ln d$), 
Corollary~\ref{cor:supp-energy-variance} controls the slope of $u_{L,A}(s)$. 
This yields the following proposition.

\begin{prop}[Uniform H\"older continuity of the canonical energy]
\label{prop:supp-uniform-holder}
There is a constant $C<\infty$, independent of $L$ and of the nonempty
subsystem $A\subseteq\Lambda_L$, such that
\begin{equation}
 |u_{L,A}(s)-u_{L,A}(t)|
 \le C E_0\sqrt{|s-t|}
 \qquad(0\le s,t\le\ln d).
 \label{eq:supp-uLA-continuity}
\end{equation}
\end{prop}

\begin{proof}
Fix a nonempty subsystem $A\subseteq\Lambda_L$.  If $H_{L,A}$ is
proportional to the identity, then $u_{L,A}$ is constant and 
the claimed continuity bound is immediate.  Otherwise, let $d_{0,L,A}$ be the ground-space
dimension and set
\begin{equation}
 s_{0,L,A}\ceqq\frac{1}{|A|}\ln d_{0,L,A}.
\end{equation}
The canonical energy density is constant on $[0,s_{0,L,A}]$.  On
$(s_{0,L,A},\ln d)$, the canonical inverse temperature
\begin{equation}
 \beta_{L,A}(s)
 \ceqq
 \left[\frac{\d u_{L,A}(s)}{\d s}\right]^{-1}
\end{equation}
is well defined, positive, and strictly decreasing.  Differentiating the
entropy and energy of the finite-volume Gibbs state gives
\begin{equation}
 \frac{\d\beta_{L,A}}{\d s}
 =-\left[
 \beta_{L,A}(s)
 \frac{\operatorname{Var}_{\beta_{L,A}(s),H_{L,A}}(H_{L,A})}{|A|}
 \right]^{-1}.
 \label{eq:supp-beta-derivative}
\end{equation}

Let $s_{\beta_0,L,A}$ denote the entropy density at which
$\beta_{L,A}(s)=\beta_0$.  For
$s\in[s_{\beta_0,L,A},\ln d)$, the variance estimate
in Eq.~\eqref{eq:supp-canonical-fluctuation} and
Eq.~\eqref{eq:supp-beta-derivative} imply
\begin{equation}
 -\frac{\d}{\d s}\beta_{L,A}(s)^2
 =
 2\left[
 \frac{\operatorname{Var}_{\beta_{L,A}(s),H_{L,A}}(H_{L,A})}{|A|}
 \right]^{-1}
 \ge\frac{2}{c_{\mathrm{var}}E_0^2}.
\end{equation}
Since $\beta_{L,A}(s)\to0$ as $s\to\ln d$, integration gives
\begin{equation}
 \beta_{L,A}(s)^{-1}
 \le E_0\sqrt{\frac{c_{\mathrm{var}}}{2(\ln d-s)}}.
 \label{eq:supp-temperature-near-infinite}
\end{equation}
For $s<s_{\beta_0,L,A}$, monotonicity instead gives
$\beta_{L,A}(s)^{-1}\le\beta_0^{-1}$.

Consequently, for every $q\in[0,\ln d]$,
\begin{align}
 u_{L,A}(\ln d)-u_{L,A}(\ln d-q)
 &=
 \int_{\ln d-q}^{\ln d}
 \frac{\d u_{L,A}(s)}{\d s}\d s\nonumber\\
 &\le E_0\sqrt{2c_{\mathrm{var}}q}+\beta_0^{-1}q.
 \label{eq:supp-uniform-continuity}
\end{align}
Finally, convexity
implies that an increment over an interval of fixed length is maximal at the
upper endpoint.  This proves Eq.~\eqref{eq:supp-uLA-continuity}.
\end{proof}

\subsubsection{Thermodynamic limit of $u(s)$}
\label{sec:supp-u-limit}
The existence of the thermodynamic pressure for translation-invariant,
absolutely summable quantum-spin interactions is
standard; see Refs.~\cite{simonStatisticalMechanicsLattice1993,bratteliOperatorAlgebrasQuantum1997,ruelleStatisticalMechanicsRigorous1999}.
We recap the argument retaining the finite-size corrections explicitly.

For a finite region $Q\Subset\Z ^D$, let
\begin{equation}
 H_Q^{\mathrm{op}}\ceqq\sum_{X\Subset Q}\Phi(X)
\end{equation}
be the open-boundary Hamiltonian and define its free-energy density by
\begin{equation}
 f_Q(\beta)\ceqq-
 \frac{1}{\beta|Q|}\ln\Tr_Q e^{-\beta H_Q^{\mathrm{op}}}.
\end{equation}
For $l\le L/2$, the following identities hold:
\begin{align}
 H_{L,Q_{L,l}(z)}
 &=\tau_{L,z}\!\left([H_{\Lambda_l}^{\mathrm{op}}]_L\right),
 \label{eq:supp-torus-open-exact}\\
 u_{L,Q_{L,l}(z)}(s)&=u_{\Lambda_l}^{\mathrm{op}}(s).
 \label{eq:supp-torus-open-u-exact}
\end{align}

For a self-adjoint operator $K$ on $\mathcal H_Q$, write
$f(K;\beta)=-(\beta|Q|)^{-1}\ln\Tr_Qe^{-\beta K}$.  If $K$ and $K'$ act on
$\mathcal H_Q$, 
the Gibbs variational principle gives
\begin{equation}
 |f(K;\beta)-f(K';\beta)|
 \le\frac{\|K-K'\|}{|Q|}.
 \label{eq:supp-free-energy-stability}
\end{equation}
To compare two side lengths $l$ and $l'$, 
choose $N$ divisible by both and
tile $\Lambda_N$ by cubes $B$ of side length $l$.  If
$K_l=\sum_BH_B^{\mathrm{op}}$, then the boundary estimate gives
\begin{equation}
 \|H_{\Lambda_N}^{\mathrm{op}}-K_l\|
 \le C_DE_0N^D\bar{r}_F(l).
\end{equation}
The partition function factorizes as
\begin{equation}
 \Tr_{\Lambda_N}e^{-\beta K_l}
 =\left(\Tr_{\Lambda_l}e^{-\beta H_{\Lambda_l}^{\mathrm{op}}}\right)^{N^D/l^D},
\end{equation}
and hence $f(K_l;\beta)=f_{\Lambda_l}(\beta)$.
Equation~\eqref{eq:supp-free-energy-stability}
therefore gives
\begin{equation}
 |f_{\Lambda_N}(\beta)-f_{\Lambda_l}(\beta)|
 \le C_DE_0\bar{r}_F(l).
\end{equation}
Repeating the argument with $l'$ and using the triangle inequality yields
\begin{equation}
 |f_{\Lambda_l}(\beta)-f_{\Lambda_{l'}}(\beta)|
 \le C_DE_0\bigl[\bar{r}_F(l)+\bar{r}_F(l')\bigr].
\end{equation}
This proves the existence of
$f(\beta)=\lim_{l\to\infty}f_{\Lambda_l}(\beta)$, with
\begin{equation}
 \sup_{\beta>0}|f_{\Lambda_l}(\beta)-f(\beta)|
 \le C_DE_0\bar{r}_F(l).
 \label{eq:supp-free-energy-limit}
\end{equation}
Since
\begin{equation}
 u_{\Lambda_l}^{\mathrm{op}}(s)
 =\sup_{\beta>0}\{f_{\Lambda_l}(\beta)+\beta^{-1}s\},
\end{equation}
taking the supremum does not increase the uniform error.  Therefore
\begin{equation}
 \sup_{s\in[0,\ln d]}|u_{\Lambda_l}^{\mathrm{op}}(s)-u(s)|
 \le C_DE_0\bar{r}_F(l),
 \label{eq:supp-u-limit-rate}
\end{equation}
where $u$ is the limiting Legendre transform of $f$.  
Equations~\eqref{eq:supp-u-limit-rate} and
\eqref{eq:supp-torus-open-u-exact} directly imply
\begin{equation}
 \sup_{s\in[0,\ln d]}
 |u_{L,Q_{L,l}(z)}(s)-u(s)|
 \le C_DE_0\bar{r}_F(l),
 \label{eq:supp-torus-cube-u-limit}
\end{equation}
which is Eq.~\eqref{eq:supp-u-finite-size-used} and also proves the existence
claimed in Eq.~\eqref{eq:supp-u-limit}.

The same estimate applies to every union of blocks from a disjoint
$l$-cube family.  Indeed, if
$\hat{A}=\bigsqcup_{A\in\Acal_{\hat{A}}}A$, then
\begin{equation}
 \left\|H_{L,\hat{A}}-
 \sum_{A\in\Acal_{\hat{A}}}H_{L,A}\right\|
 \le\sum_{A\in\Acal_{\hat{A}}}b_{\Phi,L}(A).
 \label{eq:supp-union-H-comparison}
\end{equation}
The canonical state of the Hamiltonian on the right factorizes at each
inverse temperature, so its canonical energy density is
$u_{L,Q_{L,l}(0)}(s)$.  Hence
\begin{equation}
 \sup_{s\in[0,\ln d]}
 |u_{L,\hat{A}}(s)-u(s)|
 \le C_DE_0\bar{r}_F(l).
 \label{eq:supp-union-u-limit}
\end{equation}
This is the finite-size estimate used in
Eq.~\eqref{eq:supp-coarse-single-level}.  Monotonicity and convexity pass to
the limit.  The continuity estimate in Eq.~\eqref{eq:supp-uLA-continuity}
also passes to $u$, which justifies all $O(\sqrt{\delta s})$ terms above.

\end{document}